\documentclass[11pt,reqno]{amsart}

\pdfoutput = 1

\usepackage[english]{babel}
\usepackage{amsmath,amsfonts,amssymb,amsthm}
\usepackage{amsfonts}
\usepackage{amsxtra}

\usepackage{array,dcolumn}
\usepackage{graphicx}
\usepackage[hyperfootnotes=true,
        pdffitwindow=true,
        plainpages=false,
        pdfpagelabels=true,
        pdfpagemode=UseOutlines,
        pdfpagelayout=SinglePage,
        pagebackref,
        hyperindex,]{hyperref}

\usepackage[T1]{fontenc}
\usepackage[utf8]{inputenc}
\usepackage[activate={true,nocompatibility},spacing,kerning]{microtype}
\usepackage[sc,osf]{mathpazo}
\microtypecontext{spacing=nonfrench}

 \calclayout

\newcommand{\mathsym}[1]{{}}
\newcommand{\unicode}[1]{{}}

\theoremstyle{plain}
\newtheorem{theorem}{Theorem}
\newtheorem{lemma}{Lemma}

\newtheorem{proposition}{Proposition}
\numberwithin{theorem}{section}
\numberwithin{lemma}{section}
\numberwithin{proposition}{section}

\theoremstyle{definition}

\theoremstyle{remark}
\newtheorem{remark}{Remark}

\numberwithin{remark}{section}

 \DeclareMathOperator{\Tr}{Tr}

\newcommand{\+}{\!+\!}
\newcommand{\m}{\!\\!}
\newcommand{\half}{\tfrac{1}{2}}

\renewcommand{\leq}{\leqslant}
\renewcommand{\geq}{\geqslant}
\newcommand{\E}{\mathbb{E}}
\newcommand{\abs}[1]{\lvert#1\rvert}
\def\mean#1{\left< #1 \right>}

\begin{document}
\title[Linear DE's for some random matrix resolvents]{Linear differential equations for the resolvents\\
of the classical matrix ensembles}

\author{Anas A. Rahman}
\address{School of Mathematics and Statistics, 
ARC Centre of Excellence for Mathematical
 and Statistical Frontiers,
University of Melbourne, Victoria 3010, Australia}
\email{anas.rahman@live.com.au}
\author{Peter J. Forrester}
\email{pjforr@unimelb.edu.au}

\maketitle

\begin{abstract}
The spectral density for random matrix $\beta$ ensembles can be written in terms of the average of the absolute value of the characteristic polynomial raised to the power of $\beta$, which for even $\beta$ is a polynomial of degree $\beta(N-1)$. In the cases of the classical Gaussian, Laguerre, and Jacobi weights, we show that this polynomial, and moreover the spectral density itself, can be characterised as the solution of a linear differential equation of degree $\beta+1$. This equation, and its companion for the resolvent, are given explicitly for $\beta=2$ and $4$ for all three classical cases, and also for $\beta=6$ in the Gaussian case. Known dualities for the spectral moments relating $\beta$ to $4/\beta$ then imply corresponding differential equations in the case $\beta=1$, and for the Gaussian ensemble, the case $\beta=2/3$. We apply the differential equations to give a systematic derivation of recurrences satisfied by the spectral moments and by the coefficients of their $1/N$ expansions, along with first-order differential equations for the coefficients of the $1/N$ expansions of the corresponding resolvents. We also present the form of the differential equations when scaled at the hard or soft edges.
\end{abstract}

\setcounter{equation}{0}
\section{Introduction}\label{s1}
Two complementary statistical quantities associated with the eigenvalues of a random matrix ensemble are the averages of products of characteristic polynomials and the $k$-point correlation function. For the so-called Wigner matrices (see e.g. \cite{PS11}), for which the entries (up to a possible symmetry requirement) are independently distributed, the former is a lot simpler. To illustrate this, let $X$ be any $N \times N$ real random matrix with all entries independent and having zero mean and standard deviation unity. Then it is straightforward to verify that the real symmetric matrix $G = \half (X + X^T)$ has average characteristic polynomial
\begin{align} \label{eq:intro1}
\mean{ \det (\lambda I_N - G) } = 2^{-N} H_N(\lambda),
\end{align} 
where $H_N(\lambda)$ is the $N\textsuperscript{th}$ Hermite polynomial \cite{FG06}. On the other hand, for general real Wigner matrices, the lowest order $k$-point correlation function, i.e.~the spectral density, exhibits no such simple formula.

If instead of Wigner matrices, one considers invariant ensembles with probability density function (PDF) proportional to $\abs{\det X}^{\alpha_0}\exp\left(- \sum_{l=1}^\infty \alpha_l \Tr X^l\right)$, the averaged characteristic polynomials and the $k$-point correlation function are seen to be closely related. For such invariant matrix ensembles, the eigenvalue PDF has the form
\begin{align} \label{eq:intro2}
\frac{1}{C_N} \prod_{l=1}^N w(x_l) \prod_{1 \leq j<k \leq N} \abs{x_k - x_j}^\beta,
\end{align}
where $\beta = 1,2$ or $4$ corresponds to the diagonalising matrices being unitary matrices with real, complex or real quarternion entries, respectively, and $w(x) = \abs{x}^{\alpha_0} \exp\left(- \sum_{l=1}^\infty \alpha_l x^l\right)$. With $N$ replaced by $N+1$ in \eqref{eq:intro2}, the one-point density $\rho_{(1),\beta,N+1}(x)$ is then given in terms of the absolute value of the $\beta\textsuperscript{th}$ moment of the characteristic polynomial according to
\begin{align}
\rho_{(1),\beta,N+1}(x) :&={N+1 \over C_{N+1}}\int_{\mathbb{R}^N} \mathrm{d}x_1 \cdots \mathrm{d}x_N \left.\prod_{l=1}^{N+1} w(x_l)\prod_{1 \leq j < k \leq N + 1}|x_k - x_j|^\beta\right\vert_{x_{N+1}=x} \nonumber
\\&= \frac{(N+1)C_N}{C_{N+1}} w(x) \mean{ \prod_{l=1}^N \abs{x-x_l}^\beta }, \label{eq:intro3}
\end{align}
where the average is over the PDF \eqref{eq:intro2} (i.e. the eigenvalue PDF for $N$ eigenvalues).

Let us now specialise to one of the three classical weights\footnotemark
\footnotetext{The terminology classical weight has its origin in the theory of orthogonal polynomials of a single variable. For a precise definition see e.g. \cite[\S 5.4.1]{Fo10}. When $a$ and $b$ are held constant, these three weights correspond to eigenvalue densities with two soft edges, a soft and hard edge, and two hard edges, respectively.}
\begin{align} \label{eq:intro4}
w(x) = 
\begin{cases}
e^{-x^2}, &\text{Gaussian,} \\
x^a e^{-x} \chi_{x>0}, &\text{Laguerre,} \\
x^a (1-x)^b \chi_{0<x<1}, &\text{Jacobi,}
\end{cases}
\end{align}
where $\chi_A = 1$ for $A$ true and $\chi_A = 0$ otherwise. 
In the Jacobi case, the average in \eqref{eq:intro3} is an example of a particular class of Selberg correlation integrals which have been shown to satisfy a first-order matrix linear differential equation of order $N+1$ \cite{FR12}, as well as a linear matrix recurrence relation in the parameter of the same order \cite{FI10a}.

Being of order $N+1$, the use of the linear differential equation in relation to analysing the spectral density is, per se, feasible for only small values of $N$. However, this circumstance changes dramatically when one takes into consideration that for the classical ensembles the moments of the characteristic polynomial satisfy duality relations \cite{Fo92j,Fo93c,BF97a,De08,Fo10,DL15}. The simplest of these is in the Gaussian case. With $w(x) = e^{-x^2}$ in \eqref{eq:intro2} denoted $\text{GE}_{\beta,N}$, the duality reads \cite{BF97a}
\begin{align} \label{eq:intro5}
\mean{\prod_{l=1}^N \Big( x - \sqrt{\tfrac{2}{\beta}}x_l \Big)^n }_{\text{GE}_{\beta,N}} =
\mean{ \prod_{l=1}^n (x-ix_l)^N }_{\text{GE}_{4/\beta,n}}.
\end{align}
This shows that for $\beta$ even, the average in \eqref{eq:intro3} in the Gaussian case is equal to a polynomial which satisfies a linear differential equation of order $\beta + 1$. Analogous dualities for the Laguerre and Jacobi cases \cite{Fo92j,Fo93c,Fo10} show that the same conclusion holds for those ensembles too.

The primary purpose of the present paper is to make explicit the order three for $\beta = 2$, and order five for $\beta = 4$, linear differential equations satisfied by the density \eqref{eq:intro3} in the classical cases \eqref{eq:intro4}. Another type of duality formula, relating spectral moments
\begin{equation} \label{eq:intro6}
m_k=\int_Ix^k\rho_{(1),\beta,N}(x)\,\mathrm{d}x,\quad I=\textrm{supp}\,\rho_{(1),\beta,N},\, k\in\mathbb{N},
\end{equation}
of the $\beta$ and $4/\beta$ ensembles \cite{DE05,DP12,FRW17,FLD16,WF14}, then allows for the determination of a fifth-order differential equation satisfied by \eqref{eq:intro3} with $\beta=1$. In a structural sense, our approach is applicable in all three cases for all even $\beta$, with the moments duality relation then implying the analogous characterisation for the coupling $4/\beta$. However, at a technical level there is an increase in complexity as the weight changes from Gaussian, to Laguerre, to Jacobi, and an increase in complexity as $\beta$ is increased. Hence, beyond the Jacobi weight with $\beta=4$ and $1$, the next case in this ordering of complexity is the Gaussian weight with $\beta=6$ and $\beta=2/3$. In the final subsection of Section \ref{s2}, we work this case out in detail, giving the explicit form of the seventh-order linear differential equation specifying the densities.

For the ensemble corresponding to \eqref{eq:intro2} with $\beta = 2$ and the Gaussian weight (referred to as the GUE), the fact that the density satisfies a third-order linear differential equation can be traced back to the work of Lawes and March \cite{LM79}. There the setting is that of interpreting the GUE eigenvalue PDF as the squared ground state wave function of
spinless non-interacting fermions in one dimension, in the presence of a harmonic confining potential. For an analogous result in the case of fermions in $d$ dimensions, see \cite{BM03,Fo19}. In the random matrix theory literature, this result has appeared in the work of G{\"o}tze and Tikhomirov \cite{GT05}, making use of an earlier result of Haagerup and Thorbj{\o}rnsen \cite{HT03} characterising the two-sided Laplace transform of the density in terms of a hypergeometric function. Also contained in \cite{GT05} is a third-order linear differential equation for the spectral density of the LUE (the Laguerre weight case of \eqref{eq:intro2} with $\beta = 2$); this equation can be found too in \cite{ATK11}. These characterisations were used to obtain optimal bounds for the rate of convergence to the limiting semi-circle law (for the GUE) and Marchenko-Pastur law (for the LUE). The earlier result of \cite{HT03} was used by Ledoux \cite{Le04} to deduce small deviation inequalities for the largest eigenvalue in the GUE. More recently, Kopelevitch \cite{Ko16} made use of the third-order differential equation satisfied by the spectral density for the GUE to study the $1/N$ expansion of the average of a linear statistic.

In the Gaussian case with $\beta = 1$ or $\beta = 4$, corresponding to the well-known GOE and GSE, respectively, the fifth-order homogeneous differential equations for the spectral density were derived in \cite{WF14} using a method based on known evaluation of the density in terms of Hermite polynomials \cite{AFNV00}. Very recently, third-order inhomogeneous differential equations have been derived for the same quantity, in which the inhomogeneous term involves the GUE density \cite{Na18}. Our approach, using Selberg correlation integrals and duality formulas, is different, and moreover unifies all the classical cases. Furthermore, it opens the way for the future study of applications analogous to those in \cite{GT05,Le04,Ko16,Na18}. An application in the present work will be to the derivation of difference equations for the moments $m_k$ of the spectral density with \eqref{eq:intro6} extended to include $k\in\mathbb{Z}$ when possible, supplementing results on this topic in \cite{HZ86,HT03,Le04,Le09,CMSV16a,CMSV16b,CMOS18} as well as the related studies \cite{VV08,No08,LV11,MS11,MS12,CDO18,CCC19}. Another will be to the characterisation of the soft and hard edge scaled densities for $\beta = 1,2$ and $4$ via differential equations, and at the soft edge for $\beta=6$ and $2/3$ as well.

\subsection*{Outline and key results}
In Section \ref{s2}, we derive homogeneous linear differential equations for the $\beta=1,2$, and $4$ Jacobi ensembles' densities, and inhomogeneous analogues of these for the resolvents
\begin{align} \label{eq:intro7}
W_{\beta,N}(x):=\int_I \frac{\rho_{(1),\beta,N}(\lambda)}{x-\lambda}\,\mathrm{d}\lambda,\quad I = \textrm{supp}\,\rho_{(1),\beta,N}.
\end{align}
These results are contained in Theorems \ref{T2.6} and \ref{T2.8}. Utilising a limiting procedure, one may derive from these theorems differential equations for the densities and resolvents of the $\beta=1,2$, and $4$ Gaussian and Laguerre ensembles. On the other hand, Gaussian and Laguerre ensemble differential equations can be derived by tweaking the proofs of Theorems \ref{T2.6} and \ref{T2.8}. We illustrate the first route in \S\ref{s2.2}, where Proposition \ref{P2.9} lists differential equations for the $\beta=1,2$, and $4$ Laguerre ensembles' densities and resolvents, recovering and complementing the result in the $\beta=2$ case \cite{GT05,ATK11}. The second route is demonstrated in \S\ref{s2.3}, where Proposition \ref{P2.10} lists differential equations satisfied by the densities and resolvents of the Gaussian ensemble when $\beta=2/3$ and $6$, complementing the $\beta=1,2$, and $4$ differential equations given in \cite{WF14}.

In Section \ref{s3}, we apply the differential equations of Section \ref{s2} to derive recurrence relations for related objects of interest. To be more specific,
\begin{itemize}
\item We begin \S\ref{s3.1} by presenting a third-order linear recurrence for the integer moments $m_k^{(J)}$ of the JUE eigenvalue density, which was recently derived in \cite{CMOS18} through a different method to ours. Propositions \ref{P3.2} to \ref{P3.4} then give recurrences for the $\beta=1$ and $4$ Jacobi and Laguerre ensemble moments and the $\beta=2/3$ and $6$ Gaussian ensemble moments, respectively. We highlight the fact that these recurrences are derived from Theorems \ref{T2.6} and \ref{T2.8} in a uniform way, and furthermore that our method retrieves known moment recurrences for the LUE \cite{Le04,CMSV16b}, GOE, GUE, and GSE \cite{HZ86,Le09,WF14}.

\item A nice feature of the classical ensembles is that for a particular choice of scalings, the spectral moments can be expanded in $1/N$ \cite{DE05,DP12,MRW15}. In \S\ref{s3.2}, we expand the Gaussian, Laguerre and Jacobi ensemble spectral moments according to equations \eqref{eq:rr11}--\eqref{eq:rr13} and give recurrences for the expansion coefficients $M_{k,l}^{(\,\cdot\,)}$. Propositions \ref{P3.5} to \ref{P3.8} treat the $\beta=2/3$ and $\beta=6$ Gaussian ensembles, the LUE, the LOE and LSE, and the JUE, respectively. We draw attention to the works \cite{HZ86,Le09,ND18} where equivalent recurrences were found for the GOE, GUE, GSE, and a constrained version of the LUE, mostly in the context of enumerative geometry \cite{KK03,Di03,LC09}.

\item Closely related to the expansion of the moments is the fact that, for the weights considered in this paper, and when the eigenvalue density is scaled to have finite support in the $N\rightarrow\infty$ regime, its resolvent admits a topological expansion \cite{BG13}. That is, if $c_N$ is a scaling parameter such that to leading order $\rho_{(1),\beta,N}(c_N \lambda)$ has a finite support as a function of $\lambda$, then $W_{\beta,N}(c_N x)$ has an asymptotic $1/N$-expansion of the form
\begin{align} \label{eq:intro8}
\frac{c_N}{N}W_{\beta,N}(c_N x)=\sum_{l=0}^{\infty}\frac{W_{\beta}^l(x)}{(N\sqrt{\kappa})^{l}},\quad\kappa:=\frac{\beta}{2},
\end{align}
where the $W_{\beta}^l(x)$ are independent of $N$, but dependent on $\kappa$. In \S\ref{s3.3}, we show that the $W_{\beta}^l(x)$ defined above satisfy first-order differential-difference equations, providing an alternative to the topological recursion (see e.g. \cite{EO09,FRW17}) for generating these expansion coefficients. Moreover, we observe that the differential equations for $W_{\beta}^0(x)$ are independent of $\beta$, which confirms that the differential equations of Section \ref{s2} display universal structures in the large $N$ limit. We treat the same cases as those in \S\ref{s3.2}, along with the GOE and GSE, in Propositions \ref{P3.10} to \ref{P3.15}.
\end{itemize}

In Section \ref{s4}, we study the differential equations of Section \ref{s2} when they have been scaled at the hard and soft edges. Theorem \ref{T4.1} gives differential equation characterisations of the eigenvalue densities of the classical matrix ensembles at the soft edge when $\beta\in\{2/3,1,2,4,6\}$, while Theorem \ref{T4.4} gives such characterisations at the hard edge when $\beta\in\{1,2,4\}$. These differential equations are simpler than their unscaled correspondents in Section \ref{s2}, but are of the same orders.

\setcounter{equation}{0}
\section{Differential Equations}\label{s2}
It has been commented above that homogeneous differential equations for the densities of the Gaussian orthogonal, unitary, and symplectic ensembles are known from previous works. This is similarly true of
the corresponding (inhomogeneous) differential equations for the corresponding resolvents. For future reference, we
present the results here, using \cite{WF14} as our source and thus choosing the Gaussian weight
to be $\exp\left(-N\kappa x^2/(2g)\right)$ (the coupling constant $g$ determines the length scale -- to leading order in $1/N$, the spectrum is supported on $(-2\sqrt{g},2\sqrt{g})$ -- and we recall from \eqref{eq:intro8}
that $\kappa=\frac{\beta}{2}$). 

\begin{proposition} \label{P2.1}
Define
\begin{align} \label{eq:de1}
\mathcal{D}_{\beta,N}^{(G)} =
\begin{cases}
\left(\frac{g}{N\sqrt{\kappa}}\right)^2\frac{\mathrm{d}^3}{\mathrm{d}x^3}-y_{(G)}^2\frac{\mathrm{d}}{\mathrm{d}x}+x,&\beta=2,
\\-\left(\frac{g}{N\sqrt{\kappa}}\right)^4\frac{\mathrm{d}^5}{\mathrm{d}x^5}+5\left[\half y_{(G)}^2-h\left(\frac{g}{N\sqrt{\kappa}}\right)\right]\left(\frac{g}{N\sqrt{\kappa}}\right)^2\frac{\mathrm{d}^3}{\mathrm{d}x^3}&
\\\quad-3\left(\frac{g}{N\sqrt{\kappa}}\right)^2x\frac{\mathrm{d}^2}{\mathrm{d}x^2}-\left[y_{(G)}^4-4h\left(\frac{g}{N\sqrt{\kappa}}\right)y_{(G)}^2-\left(\frac{g}{N\sqrt{\kappa}}\right)^2\right]\frac{\mathrm{d}}{\mathrm{d}x}&
\\\quad+\left[y_{(G)}^2-2h\left(\frac{g}{N\sqrt{\kappa}}\right)\right]x,&\beta=1,4,
\end{cases}
\end{align}
where $h:=\sqrt{\kappa}-1/\sqrt{\kappa}$ and $y_{(G)}=\sqrt{x^2-4g}$. Then, for $\beta=1,2$, and $4$,
\begin{align} \label{eq:de2}
\mathcal{D}_{\beta,N}^{(G)}\,\rho_{(1),\beta,N}^{(G)}(x) = 0
\end{align}
and
\begin{align} \label{eq:de3}
\mathcal{D}_{\beta,N}^{(G)}\,\frac{1}{N}W_{\beta,N}^{(G)}(x) =
\begin{cases}
2,&\beta=2,
\\2y_{(G)}^2-10h\left(\frac{g}{N\sqrt{\kappa}}\right),&\beta=1,4.
\end{cases}
\end{align}
\end{proposition}

\begin{remark} \label{R2.2}
\begin{enumerate}
\item It can be observed that for $\beta=1$ or $4$,
\begin{align}
\mathcal{D}_{\beta,N}^{(G)}&=-\left(\frac{g}{N\sqrt{\kappa}}\right)^2\left[\mathcal{D}_{2,N}^{(G)}+2x\right]\frac{\mathrm{d}^2}{\mathrm{d}x^2}+\left[y_{(G)}^2-5h\left(\frac{g}{N\sqrt{\kappa}}\right)\right]\mathcal{D}_{2,N}^{(G)} \nonumber
\\&\quad+\half y_{(G)}^2\left(\frac{g}{N\sqrt{\kappa}}\right)^2\frac{\mathrm{d}^3}{\mathrm{d}x^3}-\left[hy_{(G)}^2-\left(\frac{g}{N\sqrt{\kappa}}\right)\right]\left(\frac{g}{N\sqrt{\kappa}}\right)\frac{\mathrm{d}}{\mathrm{d}x} \nonumber
\\&\quad+3h\left(\frac{g}{N\sqrt{\kappa}}\right)x. \label{eq:de4}
\end{align}
This form of the differential operator should be compared with results derived recently in \cite{Na18}.

\item The operator for $\beta=2$ is even in $N$. In the case of $\beta=1$ and $\beta=4$, it is invariant under the mapping $(N,\kappa)\mapsto(-N\kappa,1/\kappa)$; this is consistent with a known duality \cite{DE05,WF14}. Analogous dualities for the Jacobi and Laguerre ensembles are what allow us to obtain differential equations for $\beta=1$ in the upcoming derivations. They also suggest that in certain cases, the LUE and JUE spectral moments, when scaled to be ${\rm O}(N)$, are odd functions in $N$ (see pp.~24-26 for more details).

\item The structure \eqref{eq:intro3} tells us that $p_{\beta,N+1}(x):=\rho_{(1),\beta,N+1}(x)/w(x)$ has a large $x$ expansion of the form
\begin{equation} \label{eq:de4a}
p_{\beta,N+1}(x)\overset{x\rightarrow\infty}{=}\sum_{k=0}^{\infty}c_kx^{\beta N-k}
\end{equation}
with $c_0=(N+1)C_N/C_{N+1}$ fixed and all other constant coefficients $c_k$ yet to be determined. In particular, for $\beta$ an even integer, this series terminates so that $p_{\beta,N+1}(x)$ is a polynomial of order $\beta N$. Substituting $\rho_{(1),2,N}^{(G)}(x)=e^{-Nx^2/(2g)}p_{2,N}^{(G)}(x)$ into \eqref{eq:de2} shows (choosing the GUE for simplicity)
\begin{equation*}
\left\{\left(\frac{g}{N}\right)^2\frac{\mathrm{d}^3}{\mathrm{d}x^3}-\frac{3g}{N}x\frac{\mathrm{d}^2}{\mathrm{d}x^2}+\left(2x^2+\frac{g}{N}(4N-3)\right)\frac{\mathrm{d}}{\mathrm{d}x}-4(N-1)x\right\}p_{2,N}^{(G)}(x)=0.
\end{equation*}
One can check that this differential equation admits a unique polynomial solution consistent with \eqref{eq:de4a}, and with $\{c_k\}_{k>0}$ completely determined by the choice of $c_0$. This latter feature is true of all the differential equation characterisations we will obtain for $\rho_{(1),\beta,N}(x)$. It holds true because the underlying differential-difference equation \eqref{eq:de10} is effectively a (multi-dimensional) first-order recurrence.
\end{enumerate}
\end{remark}

Proposition \ref{P2.1} will be used in Section \ref{s4}. The remainder of this section is devoted to deriving similar results for the Laguerre and Jacobi ensembles for $\beta=1,2$, and $4$, and also extending our methods to the Gaussian ensemble with $\beta=6$ and (through the aforementioned duality) $4/\beta=6$.

\subsection{The Jacobi ensemble differential equations}\label{s2.1}
Our main object of interest here is
\begin{align} \label{eq:de5}
I_{\beta,N}^{(J)}(x) := \mean{ \prod_{l=1}^N|x-x_l|^{\beta} }_{\text{JE}_{\beta,N}(a,b)},
\end{align}
where we have written ${\text{JE}_{\beta,N}(a,b)}$ to indicate that the average is taken with respect to the PDF \eqref{eq:intro2} with $w(x)=x^a(1-x)^b$.
The analogue of the duality \eqref{eq:intro5} in the Jacobi case is \cite[Ch.~13]{Fo10}
\begin{align}
\frac{S_N(a,b,\kappa)}{S_N(a+n,b,\kappa)}\mean{ \prod_{l=1}^N(x_l-x)^n }_{\text{JE}_{\beta,N}(a,b)}\, &= \,\mean{ \prod_{l=1}^n(1-xx_l)^N }_{\text{JE}_{4/\beta,n}(a',b')}, \label{eq:de6}
\\a' = \frac{1}{\kappa}(a+b+2)+N-2,&\qquad b' = -\frac{1}{\kappa}(b+n)-N, \nonumber
\end{align}

\noindent where $S_N$ is the Selberg integral \cite[Ch.~4]{Fo10}. In fact, both of these averages are known to equal \cite{Ka93}
\begin{align} \label{eq:de7}
{}_2F_1^{(\kappa)}\left(-N,(a+b+n+1)/\kappa+N-1,(a+n)/\kappa;(x)^n\right),
\end{align}
where ${}_2F_1^{(\kappa)}$ is the generalised multivariate hypergeometric function and $(x)^n$ is the $n$-tuple $(x,\ldots,x)$. Note that on the right-hand side of \eqref{eq:de6}, the parameter $b'$ is in general less than $-1$, and so the integral must be understood in the sense of analytic continuation. The average (\ref{eq:de5}) relates to the left-hand side of (\ref{eq:de6}) in the case
$n = \beta$ for even $\beta$, the latter detail being needed to remove the absolute value
sign.

A simple change of variables shows that the average on the right-hand side of \eqref{eq:de6} is given by $(-x)^{nN}\,J_{n,0}^{(N)}(1/x)$, where
\begin{align} \label{eq:de8}
J_{n,p}^{(N)}(x) := \mean{ \prod_{l=1}^n(x_l-x)^{N+\chi_{l\leq p}} }_{\text{JE}_{4/\beta,n}(a',b')}.
\end{align}
Thus for $\beta$ an even positive integer, 
\begin{equation} \label{eq:de9}
I_{\beta,N}^{(J)}(x) \propto (-x)^{\beta N}\,J_{\beta,0}^{(N)}(1/x).
\end{equation}
The significance of the generalised average (\ref{eq:de8}) is that it
satisfies the differential-difference equation \cite{Fo93}
\begin{multline*}
(n-p)E_pJ_{n,p+1}^{(N)}(x) = -(A_px+B_p)J_{n,p}^{(N)}(x)+x(x-1)\frac{\mathrm{d}}{\mathrm{d}x}J_{n,p}^{(N)}(x)+D_px(x-1)J_{n,p-1}^{(N)}(x),
\end{multline*}
\begin{align}
A_p& = (n-p)\left(a'+b'+\tfrac{2}{\kappa}(n-p-1)+2N+2\right), \nonumber
\\B_p& = (p-n)\left(a'+N+1+\tfrac{1}{\kappa}(n-p-1)\right), \nonumber
\\D_p& = p\left(\tfrac{1}{\kappa}(n-p)+N+1\right), \nonumber
\\E_p& = a'+b'+\tfrac{1}{\kappa}(2n-p-2)+N+2, \label{eq:de10}
\end{align}
later observed to be equivalent to a particular matrix differential equation \cite{FR12}.

\begin{remark} \label{R2.3}
The order $p$ elementary symmetric polynomial on $N$ variables is
\begin{equation} \label{eq:de11}
e_p(x_1,\ldots,x_N)=\sum_{1\leq j_1<j_2<\cdots<j_p\leq N}x_{j_1}\cdots x_{j_p}
\end{equation}
with the convention $e_p(x_1,\ldots,x_N)=0$ if $p>N$. It is known \cite[Ch.~4]{Fo10} that the above differential-difference equation \eqref{eq:de10} is satisfied by a broader class of functions
\begin{equation} \label{eq:de12}
\tilde{J}_{n,p,q}^{(N)}(x):=\frac{1}{\mathcal{N}_{N,p}}\mean{ \prod_{l=1}^n|x_l-x|^N\,\chi_{x_1,\ldots,x_{N-q}<x}\,e_p(x-x_1,\ldots,x-x_N) }_{\text{JE}_{4/\beta,n}(a',b')},
\end{equation}
where $\mathcal{N}_{N,p}:=\binom{N}{p}$, this being the number of terms in the sum on the right-hand side of \eqref{eq:de11}. In the case $q=N$, we see $\tilde{J}_{n,p,q}^{(N)}(x)$ is equal to $J_{n,p}^{(N)}(x)$ due to symmetry of the integrand in \eqref{eq:de8}.
\end{remark}

Our present goal is to derive a scalar differential equation for $J_{\beta,0}^{(N)}(x)$. The differential-difference equation \eqref{eq:de10} is sufficient for this, since the values $p=0,1,\ldots,\beta$ result in a closed system of equations. While this method is in principle applicable for all $\beta\in2\mathbb{N}$, the resulting equation will be of order $\beta + 1$, so we address only $\beta=2$ and $4$ in this paper.

\begin{lemma} \label{L2.4}
With $p=0$ and $n=\beta=2$, the function $J_{2,0}^{(N)}(x)$ satisfies the differential equation
\begin{multline} \label{eq:de13}
0=x^2(x-1)^2\frac{\mathrm{d}^3}{\mathrm{d}x^3}J_{2,0}^{(N)}(x)-\left[3C_2(x)-2(1-2x)\right]x(x-1)\frac{\mathrm{d}^2}{\mathrm{d}x^2}J_{2,0}^{(N)}(x)
\\+\left[\left((a+N)(1+4N)+4+N\right)x(x-1)+\left(2C_2(x)-3(1-2x)\right)C_2(x)\right]\frac{\mathrm{d}}{\mathrm{d}x}J_{2,0}^{(N)}(x)
\\-2N(a+N)\left[2C_2(x)-3(1-2x)\right]J_{2,0}^{(N)}(x),
\end{multline}
where $C_2(x)=(a+2N)(x-1)-b-2x$.
\end{lemma}
\begin{proof}
Setting $n=2$ and taking $p=0,1,2$ in \eqref{eq:de10}, we obtain the matrix differential equation
\begin{align} \label{eq:de14}
\frac{\mathrm{d}}{\mathrm{d}x}\begin{bmatrix}J_{2,0}^{(N)}(x)\\J_{2,1}^{(N)}(x)\\J_{2,2}^{(N)}(x)\end{bmatrix}=\begin{bmatrix}\frac{A_0x+B_0}{x(x-1)}&\frac{2E_0}{x(x-1)}&0\\-D_1&\frac{A_1x+B_1}{x(x-1)}&\frac{E_1}{x(x-1)}\\0&-D_2&0 \end{bmatrix}\begin{bmatrix}J_{2,0}^{(N)}(x)\\J_{2,1}^{(N)}(x)\\J_{2,2}^{(N)}(x)\end{bmatrix}.
\end{align}
The second row gives an expression for $J_{2,2}^{(N)}(x)$ which transforms the third row into a differential equation involving only $J_{2,0}^{(N)}(x)$ and $J_{2,1}^{(N)}(x)$. This equation further transforms into an equation for just $J_{2,0}^{(N)}(x)$ upon substitution of the expression for $J_{2,1}^{(N)}(x)$ drawn from the first row:
\begin{align}
0&=x^2(x-1)^2\frac{\mathrm{d}^3}{\mathrm{d}x^3}J_{2,0}^{(N)}(x)-\left[\tilde{C}_{2,0}(x)+\tilde{C}_{2,1}(x)+1-2x\right]x(x-1)\frac{\mathrm{d}^2}{\mathrm{d}x^2}J_{2,0}^{(N)}(x) \nonumber
\\&\quad+\left[(D_2E_1+2D_1E_0-A_1-2A_0+2)x(x-1)+\tilde{C}_{2,0}(x)\tilde{C}_{2,1}(x)\right]\frac{\mathrm{d}}{\mathrm{d}x}J_{2,0}^{(N)}(x) \nonumber
\\&\quad+\left[A_0\tilde{C}_{2,1}(x)+(A_1-D_2E_1)(A_0x+B_0)-2D_1E_0(1-2x)\right]J_{2,0}^{(N)}(x), \label{eq:de15}
\end{align}
where $\tilde{C}_{2,p}(x)=(A_p-2)x+B_p+1$ with $n=\beta=2$. Substituting the appropriate values for the constants $A_p,B_p,D_p$ and $E_p$ gives the claimed result.
\end{proof}

\begin{lemma} \label{L2.5}
With $p=0$ and $n=\beta=4$, the function $J_{4,0}^{(N)}(x)$ satisfies the differential equation with polynomial coefficients
\begin{align} \label{eq:de16}
0&=4x^4(x-1)^4\frac{\mathrm{d}^5}{\mathrm{d}x^5}J_{4,0}^{(N)}(x)- 20\left[(a+4N)(x-1)-b-2x\right] x^3(x-1)^3\frac{\mathrm{d}^4}{\mathrm{d}x^4}J_{4,0}^{(N)}(x) \nonumber
\\&\quad+\left[5(a+4N)^2-5a(a-2)-12\right] x^3(x-1)^3\frac{\mathrm{d}^3}{\mathrm{d}x^3}J_{4,0}^{(N)}(x) + \cdots
\end{align}
where the (lengthier) specific forms of the coefficients of the lower order derivatives have been suppressed. (One may obtain the full expression by inverting the proof of Theorem \ref{T2.8}. In fact, this is the most efficient method of obtaining the full expression using computer algebra.)
\end{lemma}
\begin{proof}
Like the preceding proof, setting $n=4$ in \eqref{eq:de10} yields the matrix differential equation
\begin{align} \label{eq:de17}
\frac{\mathrm{d}}{\mathrm{d}x}\begin{bmatrix}J_{4,0}^{(N)}(x)\\J_{4,1}^{(N)}(x)\\J_{4,2}^{(N)}(x)\\J_{4,3}^{(N)}(x)\\J_{4,4}^{(N)}(x)\end{bmatrix}=\begin{bmatrix}\frac{A_0x+B_0}{x(x-1)}&\frac{4E_0}{x(x-1)}&0&0&0\\-D_1&\frac{A_1x+B_1}{x(x-1)}&\frac{3E_1}{x(x-1)}&0&0\\0&-D_2&\frac{A_2x+B_2}{x(x-1)}&\frac{2E_2}{x(x-1)}&0\\0&0&-D_3&\frac{A_3x+B_3}{x(x-1)}&\frac{E_3}{x(x-1)}\\0&0&0&-D_4&0\end{bmatrix}\begin{bmatrix}J_{4,0}^{(N)}(x)\\J_{4,1}^{(N)}(x)\\J_{4,2}^{(N)}(x)\\J_{4,3}^{(N)}(x)\\J_{4,4}^{(N)}(x)\end{bmatrix}.
\end{align}
For $1\leq p\leq4$, the $p\textsuperscript{th}$ row gives an expression for $J_{4,p}^{(N)}(x)$ in terms of $\frac{\mathrm{d}}{\mathrm{d}x}J_{4,p-1}^{(N)}(x)$ and $J_{4,k}^{(N)}(x)$ for $k<p$. Substituting these expressions (in the order of decreasing $p$) into the differential equation corresponding to the fifth row yields a fifth-order differential equation for $J_{4,0}^{(N)}(x)$ similar to that of Lemma \ref{L2.4}.
\end{proof}

Now, we may easily obtain differential equations for $I_{\beta,N}^{(J)}(x)$ for $\beta=2$ and $4$, which in turn give us differential equations for $\rho_{(1),\beta,N}^{(J)}(x)$ for the same $\beta$ values.
\begin{theorem} \label{T2.6}
Define
\begin{align} \label{eq:de18}
\mathcal{D}_{2,N}^{(J)}&=x^3(1-x)^3\frac{\mathrm{d}^3}{\mathrm{d}x^3}+4(1-2x)x^2(1-x)^2\frac{\mathrm{d}^2}{\mathrm{d}x^2} \nonumber
\\&\quad+\left[(a+b+2N)^2-14\right]x^2(1-x)^2\frac{\mathrm{d}}{\mathrm{d}x}-\left[a^2(1-x)+b^2x-2\right]x(1-x)\frac{\mathrm{d}}{\mathrm{d}x} \nonumber
\\&\quad+\tfrac{1}{2}\left[(a+b+2N)^2-4\right](1-2x)x(1-x)+\tfrac{3}{2}\left[a^2-b^2\right]x(1-x) \nonumber
\\&\quad-a^2(1-x)+b^2x.
\end{align}
Then,
\begin{align} \label{eq:de19}
\mathcal{D}_{2,N}^{(J)}\,\rho_{(1),2,N}^{(J)}(x)=0
\end{align}
and
\begin{align} \label{eq:de20}
\mathcal{D}_{2,N}^{(J)}\,\frac{1}{N}W_{2,N}^{(J)}(x)=(a+b+N)(a(1-x)+bx).
\end{align}
\end{theorem}

\begin{proof}
We change variables $x\mapsto1/x$ in \eqref{eq:de13} to obtain a differential equation for $J_{2,0}^{(N)}(1/x)$, using the fact that $\frac{\mathrm{d}}{\mathrm{d}(1/x)}=-x^2\frac{\mathrm{d}}{\mathrm{d}x}$. Since this differential equation is homogeneous, we can ignore constants of proportionality and substitute in
\begin{align*}
J_{2,0}^{(N)}(1/x)=x^{-a-2N}(1-x)^{-b}\rho_{(1),2,N+1}^{(J)}(x)
\end{align*}
according to \eqref{eq:intro3} and \eqref{eq:de9}. Repeatedly applying the product rule and then replacing $N+1$ by $N$ gives \eqref{eq:de19}. Applying the Stieltjes transform to this equation term by term (see Appendix \ref{A}) and substituting in the values of the spectral moments $m_1^{(J)}$ and $m_2^{(J)}$ with $\beta=2$ \cite{MRW15} yields \eqref{eq:de20}.
\end{proof}
This differential equation is a generalisation of that given in \cite[\S4]{GKR05} for the $\beta=2$ Legendre ensemble, which is the Jacobi ensemble with $a=b=0$. It lowers the order by one relative to the fourth-order differential equation
for $\rho_{(1),2,N}^{(J)}(x)$ given in the recent work \cite[Prop.~6.7]{CMOS18}.

\begin{remark} \label{R2.7}
It has been observed in \cite[\S 3.3]{FT18} that the third-order differential equation for the spectral density in the Gaussian and Laguerre cases for $\beta = 2$ are equivalent to particular $\sigma$ Painlev\'e equations
implicit from the characterisation of gap probabilities in those cases (see e.g.~\cite[Ch.~8]{Fo10}). An analogous result holds true for the differential equation (\ref{eq:de19}). Thus, with
$v_1 = v_3 = N + (a+b)/2$, $v_2 = (a+b)/2$, $v_4 = (b-a)/2$, in studying the gap for the interval $(0,s)$ one encounters the nonlinear equation \cite[Eq.~(8.76)]{Fo10}
\begin{eqnarray*}\label{8.sf}
&&
(t(1-t)f'')^2 - 4t(1-t)(f')^3 + 4(1-2t)(f')^2f + 4f'f^2 - 4 f^2 v_1^2
\nonumber \\
&&
\quad + (f')^2\Big ( 4tv_1^2(1-t) - (v_2 - v_4)^2 - 4tv_2 v_4 \Big )
+ 4ff'(-v_1^2+2tv_1^2+v_2v_4)  = 0,
\end{eqnarray*}
subject to the boundary condition
$
f(t) \underset{t\rightarrow0^+}{\sim} - \xi t(1-t) \rho_{(1),2,N}^{(J)}(t)$. Substituting this boundary condition for $f$ and equating terms of order $\xi^2$ shows that $u(t):=  t(1-t) \rho_{(1),2,N}^{(J)}(t)$ satisfies the
second-order nonlinear differential equation
\begin{multline} \label{eq:de21}
(t(1-t) u''(t))^2 - 4 (u(t))^2 v_1^2 + (u'(t))^2 (4tv_1^2 (1 - t) - (v_2 - v_4)^2 - 4 t v_2 v_4)  \\
+ 4 u(t) u'(t) (- v_1^2 + 2 t v_1^2 + v_2 v_4) = 0.
\end{multline}
Differentiating this, and simplifying, gives a third-order linear differential equation that agrees with \eqref{eq:de19}.
\end{remark}

\begin{theorem} \label{T2.8}
Recalling that $\kappa=\beta/2$, let
\begin{align*}
a_{\beta}:=\frac{a}{\kappa-1},\quad b_{\beta}:=\frac{b}{\kappa-1},\quad N_{\beta}:=(\kappa-1)N
\end{align*}
so that $(a_{4},b_{4},N_{4})=(a,b,N)$ and $(a_{1},b_{1},N_{1})=(-2a,-2b,-N/2)$. For $\beta=1$ or $4$, define
\begin{align} \label{eq:de22}
\mathcal{D}_{\beta,N}^{(J)}&=4x^5(1-x)^5\frac{\mathrm{d}^5}{\mathrm{d}x^5}+40(1-2x)x^4(1-x)^4\frac{\mathrm{d}^4}{\mathrm{d}x^4}+\left[5\tilde{c}^2-493\right]x^4(1-x)^4\frac{\mathrm{d}^3}{\mathrm{d}x^3}\nonumber
\\&\quad-\left[5f_+(x;\tilde{a},\tilde{b})-88\right]x^3(1-x)^3\frac{\mathrm{d}^3}{\mathrm{d}x^3}+41\left[\tilde{a}-\tilde{b}\right]x^3(1-x)^3\frac{\mathrm{d}^2}{\mathrm{d}x^2}\nonumber
\\&\quad+\left[19\tilde{c}^2-539\right](1-2x)x^3(1-x)^3\frac{\mathrm{d}^2}{\mathrm{d}x^2}-22f_-(x;\tilde{a},\tilde{b})x^2(1-x)^2\frac{\mathrm{d}^2}{\mathrm{d}x^2}\nonumber
\\&\quad+16(1-2x)x^2(1-x)^2\frac{\mathrm{d}^2}{\mathrm{d}x^2}+\left[\tilde{c}^4-64\tilde{c}^2+719\right]x^3(1-x)^3\frac{\mathrm{d}}{\mathrm{d}x}\nonumber
\\&\quad-\left[(\tilde{c}^2-45)(\tilde{a}+\tilde{b}-6)+(\tilde{a}-\tilde{b})^2-248\right]x^2(1-x)^2\frac{\mathrm{d}}{\mathrm{d}x}\nonumber
\\&\quad-\left[(\tilde{c}^2-37)(\tilde{a}-\tilde{b})\right](1-2x)x^2(1-x)^2\frac{\mathrm{d}}{\mathrm{d}x}\nonumber
\\&\quad+\left[f_+(x;\tilde{a}^2,\tilde{b}^2)-14f_+(x;\tilde{a},\tilde{b})-16\right]x(1-x)\frac{\mathrm{d}}{\mathrm{d}x}\nonumber
\\&\quad+\tfrac{1}{2}\left[5(\tilde{c}^2-9)(\tilde{a}-\tilde{b})\right]x^2(1-x)^2+\tfrac{1}{2}(\tilde{c}^2-9)^2(1-2x)x^2(1-x)^2 \nonumber
\\&\quad-\tfrac{1}{2}\left[(3\tilde{c}^2-35)f_-(x;\tilde{a},\tilde{b})+\tfrac{7}{2}(\tilde{a}^2-\tilde{b}^2)+4(\tilde{a}-\tilde{b})\right]x(1-x) \nonumber
\\&\quad-\tfrac{1}{2}\left[4\tilde{c}^2-36+\tfrac{3}{2}(\tilde{a}-\tilde{b})^2\right](1-2x)x(1-x)+f_-(x;\tilde{a}^2,\tilde{b}^2),
\end{align}
where
\begin{align} \label{eq:de23}
\tilde{a}=a_{\beta}(a_{\beta}-2),\quad\tilde{b}&=b_{\beta}(b_{\beta}-2),\quad \tilde{c}=a_{\beta}+b_{\beta}+4N_{\beta}-1,\nonumber
\\f_{\pm}(x;\tilde{a},\tilde{b})&=\tilde{a}(1-x)\pm\tilde{b}x.
\end{align}
Then, for $\beta=1$ and $4$,
\begin{align} \label{eq:de24}
\mathcal{D}_{\beta,N}^{(J)}\,\rho_{(1),\beta,N}^{(J)}(x)=0
\end{align}
and
\begin{multline} \label{eq:de25}
\mathcal{D}_{\beta,N}^{(J)}\,\frac{1}{N}W_{\beta,N}^{(J)}(x)=(\tilde{c}-2N_{\beta})x(1-x)\Big[(a_{\beta}+b_{\beta})(a_{\beta}+b_{\beta}-2)(a_{\beta}(1-x)+b_{\beta}x)
\\+4N_{\beta}(\tilde{c}-2N_{\beta})(2a_{\beta}(1-x)+2b_{\beta}x-1)-a_{\beta}b_{\beta}(a_{\beta}+b_{\beta}-6)-4(a_{\beta}+b_{\beta}-1)\Big]
\\-(\tilde{c}-2N_{\beta})\left(a_{\beta}(a_{\beta}-2)^2(1-x)^2+b_{\beta}(b_{\beta}-2)^2x^2\right).
\end{multline}
\end{theorem}

\begin{proof}
The proof is done in four steps. First, to see that equation \eqref{eq:de24} holds for $\beta=4$, one undertakes the same steps as in the proof of Theorem \ref{T2.6}. That is, change variables $x\mapsto 1/x$ in \eqref{eq:de16}, substitute in
\begin{align*}
J_{4,0}^{(N)}(1/x)=x^{-a-4N}(1-x)^{-b}\rho_{(1),4,N+1}^{(J)}(x),
\end{align*}
and then replace $N+1$ by $N$. For the second step, apply the Stieltjes transform according to Appendix \ref{A} and substitute in the values of the spectral moments $m_1^{(J)}$ to $m_4^{(J)}$ \cite{MOPS,MRW15} to obtain \eqref{eq:de25} for $\beta=4$. The third step proves \eqref{eq:de25} for $\beta=1$ by employing the known duality \cite{DP12,FRW17,FLD16}
\begin{equation*}
W_{4,N}^{(J)}(x;a,b)=W_{1,-N/2}^{(J)}(-2x;-2a,-2b),
\end{equation*}
which leads us to formulate \eqref{eq:de25} in terms of the $(a_{\beta},b_{\beta},N_{\beta})$ parameters. Since this step essentially redefines constants, applying the inverse Stieltjes transform to this result returns \eqref{eq:de24} with the new constants.
\end{proof}

In the way that \eqref{eq:de22} is presented, all of its dependencies on $a,b$ and $N$ are captured by $\tilde{a},\tilde{b}$ and $\tilde{c}$. Moreover, it can be seen that this operator is invariant under the symmetry $(x,a,b)\leftrightarrow(1-x,b,a)$, which is a property of $\rho_{(1),\beta,N}^{(J)}(x)$. Equations \eqref{eq:de19} and \eqref{eq:de24} have been checked for $N=1,2$ and to hold in the large $N$ limit, while \eqref{eq:de20} and \eqref{eq:de25} have been checked to be consistent with expressions for the topological expansion coefficients $W_{\beta}^{(J),0}(x),\ldots, W_{\beta}^{(J),4}(x)$ generated via the method presented in \cite{FRW17}. It should be noted that while the moments $m_1^{(J)}$ to $m_4^{(J)}$ of the Jacobi $\beta$ ensemble's eigenvalue density are complicated rational functions of $a,b$, and $N$, the right-hand side of \eqref{eq:de20} and \eqref{eq:de25} are relatively simple polynomials, even though they are linear combinations of these moments.

\subsection{The Laguerre ensemble differential equations}\label{s2.2}
The eigenvalue density $\rho_{(1),\beta,N}^{(L)}(x)$ of the Laguerre $\beta$ ensemble can be obtained from that of the Jacobi $\beta$ ensemble via a limiting procedure: From \eqref{eq:intro3}, upon straightforward scaling, we see
\begin{equation} \label{eq:de26}
\rho_{(1),\beta,N+1}^{(L)}(x) = \lim_{b\rightarrow\infty}b^{(N+3)N\kappa+(N+2)a+N+1}\,\rho_{(1),\beta,N+1}^{(J)}\Big (\frac{x}{b} \Big ).
\end{equation}
This fact allows one to transform differential equations satisfied by $\rho_{(1),\beta,N}^{(J)}(x)$ into analogues satisfied by $\rho_{(1),\beta,N}^{(L)}(x)$, which will be presented in a moment.

As an aside, suppose that one would like to obtain differential equations for $\rho_{(1),\beta,N}^{(L)}(x)$ without prior knowledge of the analogous differential equations for $\rho_{(1),\beta,N}^{(J)}(x)$, i.e. if the results of \S\ref{s2.1} were not available, or if one were interested in ensembles with $\beta\notin\{1,2,4\}$. Then, it is actually more efficient to circumvent computation of the differential equations for $\rho_{(1),\beta,N}^{(J)}(x)$ and use the aforementioned limiting procedure indirectly. That is, let
$$
L_{\beta,p}^{(N)}(x) =\lim_{b\rightarrow\infty}\left(-\tfrac{x}{b}\right)^{\beta N+p}J_{\beta,p}^{(N)}(b/x), \qquad
 \tilde{L}_{\beta,p}^{(N)}(x) =x^ae^{-x}\,L_{\beta,p}^{(N)}(x)
$$
so that by \eqref{eq:intro3}, $\rho_{(1),\beta,N+1}^{(L)}(x)$ is proportional to $\tilde{L}_{\beta,0}^{(N)}(x)$ when $\beta$ is a positive even integer. Equation \eqref{eq:de10} gives a differential-difference equation for $\left(-\tfrac{x}{b}\right)^{\beta N+p}J_{\beta,p}^{(N)}(b/x)$ and taking the limit $b\rightarrow\infty$ thus gives a differential-difference equation for $L_{\beta,p}^{(N)}(x)$. Substituting $L_{\beta,p}^{(N)}(x)=x^{-a}e^x\tilde{L}_{\beta,p}^{(N)}(x)$ then gives a differential-difference equation for $\tilde{L}_{\beta,p}^{(N)}(x)$. These equations with $p=0,1,\ldots,\beta$ are equivalent to matrix differential equations not unlike \eqref{eq:de14} and \eqref{eq:de17}. However, having taken the limit $b\rightarrow\infty$, we ensure that these new matrix differential equations are simpler, and have the added benefit of simplifying down to scalar differential equations for $\rho_{(1),\beta,N+1}^{(L)}(x)$ rather than for auxiliary functions. One may then apply the Stieltjes transform to obtain differential equations for the resolvents, and use appropriate $\beta\leftrightarrow4/\beta$ moment dualities \cite{DE05,FRW17,FLD16} to obtain mirror differential equations like those seen in Theorem \ref{T2.8}.

Since we are interested in $\rho_{(1),\beta,N}^{(L)}(x)$ with $\beta\in\{1,2,4\}$ and we have differential equations for $\rho_{(1),\beta,N}^{(J)}(x)$ for these $\beta$ values, we use the more direct approach to give the following proposition.

\begin{proposition} \label{P2.9}
Retaining the definitions of $a_{\beta},N_{\beta}$, and $\tilde{a}$ from Theorem \ref{T2.8}, define
\begin{align} \label{eq:de27}
\mathcal{D}_{2,N}^{(L)}&=x^3\frac{\mathrm{d}^3}{\mathrm{d}x^3}+4x^2\frac{\mathrm{d}^2}{\mathrm{d}x^2}-\left[x^2-2(a+2N)x+a^2-2\right]x\frac{\mathrm{d}}{\mathrm{d}x}+\left[(a+2N)x-a^2\right]
\end{align}
and for $\beta=1$ or $4$,
\begin{align} \label{eq:de28}
\mathcal{D}_{\beta,N}^{(L)}&=4x^5\frac{\mathrm{d}^5}{\mathrm{d}x^5}+40x^4\frac{\mathrm{d}^4}{\mathrm{d}x^4}-\left[5\left(\frac{x}{\kappa-1}\right)^2-10\left(a_{\beta}+4N_{\beta}\right)\frac{x}{\kappa-1}+5\tilde{a}-88\right]x^3\frac{\mathrm{d}^3}{\mathrm{d}x^3}\nonumber
\\&\quad-\left[16\left(\frac{x}{\kappa-1}\right)^2-38\left(a_{\beta}+4N_{\beta}\right)\frac{x}{\kappa-1}+22\tilde{a}-16\right]x^2\frac{\mathrm{d}^2}{\mathrm{d}x^2}\nonumber
\\&\quad+\left[\left(\frac{x}{\kappa-1}\right)^2-4\left(a_{\beta}+4N_{\beta}\right)\left(\frac{x}{\kappa-1}\right)+2\left(2\left(a_{\beta}+4N_{\beta}\right)^2+\tilde{a}-2\right)\right]\frac{x^3}{(\kappa-1)^2}\frac{\mathrm{d}}{\mathrm{d}x}\nonumber
\\&\quad-\left[4(\tilde{a}-3)\left(a_{\beta}+4N_{\beta}\right)\frac{x}{\kappa-1}-\tilde{a}^2+14\tilde{a}+16\right]x\frac{\mathrm{d}}{\mathrm{d}x}-\left(a_{\beta}+4N_{\beta}\right)\left(\frac{x}{\kappa-1}\right)^3\nonumber
\\&\quad+\left(2\left(a_{\beta}+4N_{\beta}\right)^2+\tilde{a}\right)\left(\frac{x}{\kappa-1}\right)^2-(3\tilde{a}+4)\left(a_{\beta}+4N_{\beta}\right)\frac{x}{\kappa-1}+\tilde{a}^2.
\end{align}
Then, for $\beta=1,2$, and $4$,
\begin{equation} \label{eq:de29}
\mathcal{D}_{\beta,N}^{(L)}\,\rho_{(1),\beta,N}^{(L)}(x)=0
\end{equation}
and
\begin{equation} \label{eq:de30}
\mathcal{D}_{\beta,N}^{(L)}\,\frac{1}{N}W_{\beta,N}^{(L)}(x)=
\begin{cases}
(x+a),&\beta=2,
\\ \frac{4}{\kappa-1}\left[2\left(\frac{x}{\kappa-1}\right)^2+(2a_{\beta}-1)\frac{x}{\kappa-1}\right]N_{\beta}&
\\\quad-\frac{1}{\kappa-1}\left[\left(\frac{x}{\kappa-1}\right)^3-(a_{\beta}+2)\left(\frac{x}{\kappa-1}\right)^2\right]&
\\\quad+\frac{1}{\kappa-1}\left[(a_{\beta}^2+4a_{\beta}-4)\frac{x}{\kappa-1}-a_{\beta}(a_{\beta}-2)^2\right],&\beta=1,4.
\end{cases}
\end{equation}
\end{proposition}
\begin{proof}
To obtain \eqref{eq:de29} from \eqref{eq:de19} and \eqref{eq:de24}, change variables $x\mapsto x/b$, multiply both sides by $b^{(N+2)(N-1)\kappa+(N+1)a+N}$ according to \eqref{eq:de26}, and then take the limit $b\rightarrow\infty$. This is equivalent to changing variables $x\mapsto x/b$, extracting terms of leading order in $b$, then rewriting $\rho_{(1),\beta,N}^{(J)}$ as $\rho_{(1),\beta,N}^{(L)}$.

To obtain \eqref{eq:de30}, apply the same prescription to \eqref{eq:de20} and \eqref{eq:de25}. Alternatively, one may apply the Stieltjes transform to \eqref{eq:de29}. This is considerably easier than in the Jacobi case (see Appendix \ref{A}), since
\begin{align*}
\int_0^{\infty}\frac{x^n}{s-x}\frac{\mathrm{d}^n}{\mathrm{d}x^n}\rho_{(1),\beta,N}^{(L)}(x)\,\mathrm{d}x=s^n\frac{\mathrm{d}^n}{\mathrm{d}s^n}W_{\beta,N}^{(L)}(s)
\end{align*}
can be computed through repeated integration by parts using the fact that the boundary terms vanish at all stages: For $a>-1$ real and $m>n$ non-negative integers, $\frac{x^m}{s-x}\frac{\mathrm{d}^n}{\mathrm{d}x^n}\rho_{(1),\beta,N}^{(L)}(x)$ has a factor of $x^{a+m-n}$ which dominates at $x=0$ and a factor of $e^{-x}$ which dominates as $x\rightarrow\infty$.
\end{proof}
Equation \eqref{eq:de29} has been checked for $N=1$ and $2$ using computer algebra. From inspection, it seems that the natural variable of \eqref{eq:de29} and \eqref{eq:de30} is $x/(\kappa-1)$, which is the limit $\lim_{b\rightarrow\infty}b_{\beta}\left(\tfrac{x}{b}\right)$. This is in keeping with the duality \cite{DE05,FRW17}
\begin{equation*}
W_{\beta,N}^{(L)}(x;a)=W_{4/\beta,-\kappa N}^{(L)}(-x/\kappa;-a/\kappa).
\end{equation*}
Strictly speaking, changing variables to $y=x/(\kappa-1)$ is not natural and would be counterproductive since the corresponding weight $\left[(\kappa-1)y\right]^ae^{(1-\kappa)y}$ vanishes at $\kappa=1$ and has different support depending on whether $\kappa<1$ or $\kappa>1$.

\subsection{Additional Gaussian ensemble differential equations}\label{s2.3}
The previous subsection contains a discussion on how one would obtain differential equations for the densities and resolvents of Laguerre ensembles with even $\beta$ without knowledge of differential equations for the corresponding Jacobi $\beta$ ensembles' densities and resolvents. We now elucidate those ideas by explicitly applying them to the Gaussian ensembles with $\beta=6$ and consequently, by duality,  $\beta=2/3$. Indeed, since we haven't investigated these $\beta$ values in the Jacobi case, we cannot immediately apply the direct limiting approach used in the proof of Proposition \ref{P2.9} and it is in fact more efficient to instead use the method presented below.

Like in the Jacobi and Laguerre cases, our initial focus is the average
\begin{equation} \label{eq:de31}
I_{\beta,N}^{(G)}(x):=\left\langle\prod_{l=1}^N|x-x_l|^{\beta}\right\rangle_{\text{GE}_{\beta,N}}.
\end{equation}
Replacing $x$ by $\sqrt{\tfrac{2}{\beta}}x$ in duality \eqref{eq:intro5} and then factoring $(-\mathrm{i})^{nN}$ from the right-hand side shows that for even $\beta$, $I_{\beta,N}^{(G)}(x)$ is proportional to $G_{\beta,0}^{(N)}(x)$ where
\begin{equation} \label{eq:de32}
G_{n,p}^{(N)}(x):=(-\mathrm{i})^{nN+p}\left\langle\prod_{l=1}^n\left(x_l+\mathrm{i}\sqrt{\tfrac{2}{\beta}}x\right)^{N+\chi_{l\leq p}}\right\rangle_{\text{GE}_{4/\beta,n}},\quad 0\leq p\leq n.
\end{equation}
Setting $a'=b'=L$ and changing variables $x_l\mapsto\tfrac{1}{2}\left(1+\tfrac{x_l}{\sqrt{L}}\right)$ in $J_{n,p}^{(N)}(x)$ \eqref{eq:de8}, we see that
\begin{multline} \label{eq:de33}
G_{n,p}^{(N)}(x)=\lim_{L\rightarrow\infty}4^{nL}(-2\mathrm{i}\sqrt{L})^{nN+p}(2\sqrt{L})^{n(n-1)/\kappa+n}\left.J_{n,p}^{(N)}\left(\tfrac{1}{2}\left(1-\mathrm{i}\sqrt{\tfrac{2}{\beta L}}x\right)\right)\right|_{a'=b'=L}.
\end{multline}
Thus, equation \eqref{eq:de10} simplifies to a differential-difference equation for $G_{n,p}^{(N)}(x)$,
\begin{multline} \label{eq:de34}
(n-p)G_{n,p+1}^{(N)}(x)=\tfrac{(n-p)}{\sqrt{\kappa}}xG_{n,p}^{(N)}(x)-\tfrac{\sqrt{\kappa}}{2}\frac{\mathrm{d}}{\mathrm{d}x}G_{n,p}^{(N)}(x)+\tfrac{p}{2}\left(\tfrac{1}{\kappa}(n-p)+N+1\right)G_{n,p-1}^{(N)}(x)
\end{multline}
(cf.~\cite[Eq.~(5.5)]{FT19}).
Taking the $L\rightarrow\infty$ limit early has already yielded a simpler equation than \eqref{eq:de10}. However, we can take this one step further by defining
\begin{equation} \label{eq:de35}
\tilde{G}_{n,p}^{(N)}(x)=e^{-x^2}G_{n,p}^{(N)}(x)
\end{equation}
so that $\rho_{(1),\beta,N+1}^{(G)}(x)$ is proportional to $\tilde{G}_{\beta,0}^{(N)}(x)$. It is then easy to obtain
\begin{multline} \label{eq:de36}
\frac{\mathrm{d}}{\mathrm{d}x}\tilde{G}_{n,p}^{(N)}(x)=\tfrac{p}{\sqrt{\kappa}}\left(\tfrac{1}{\kappa}(n-p)+N+1\right)\tilde{G}_{n,p-1}^{(N)}(x)
\\+2\left(\tfrac{1}{\kappa}(n-p)-1\right)x\tilde{G}_{n,p}^{(N)}(x)-\tfrac{2}{\sqrt{\kappa}}(n-p)\tilde{G}_{n,p+1}^{(N)}(x).
\end{multline}
With $n=\beta\in2\mathbb{N}$ and $p=0,1,\ldots,n$, this is equivalent to a matrix differential equation which is moreover equivalent to a scalar differential equation for $\rho_{(1),\beta,N+1}^{(G)}(x)$.

\begin{proposition} \label{P2.10}
For $\beta=2/3$ or $6$, define
\begin{align} \label{eq:de37}
\mathcal{D}_{\beta,N}^{(G)}&=81(\kappa-1)^{7/2}\frac{\mathrm{d}^7}{\mathrm{d}x^7}+1008\left(3N_{\beta}-\frac{2x^2}{\kappa-1}+2\right)(\kappa-1)^{5/2}\frac{\mathrm{d}^5}{\mathrm{d}x^5} \nonumber
\\&\quad+2016x(\kappa-1)^{3/2}\frac{\mathrm{d}^4}{\mathrm{d}x^4} \nonumber
\\&\quad+64\left(21N_{\beta}-\frac{14x^2}{\kappa-1}+5\right)\left(21N_{\beta}-\frac{14x^2}{\kappa-1}+23\right)(\kappa-1)^{3/2}\frac{\mathrm{d}^3}{\mathrm{d}x^3} \nonumber
\\&\quad+9984\left(3N_{\beta}-\frac{2x^2}{\kappa-1}+2\right)x(\kappa-1)^{1/2}\frac{\mathrm{d}^2}{\mathrm{d}x^2} \nonumber
\\&\quad+256\Bigg[54N_{\beta}\left(4N_{\beta}^2+8N_{\beta}+3\right)-(432N_{\beta}^2+576N_{\beta}+57)\frac{x^2}{\kappa-1} \nonumber
\\&\quad+96(3N_{\beta}+2)\frac{x^4}{(\kappa-1)^2}-\frac{64x^6}{(\kappa-1)^3}-20\Bigg](\kappa-1)^{1/2}\frac{\mathrm{d}}{\mathrm{d}x} \nonumber
\\&\quad+256\left(144N_{\beta}^2+192N_{\beta}-64(3N_{\beta}+2)\frac{x^2}{\kappa-1}+\frac{64x^4}{(\kappa-1)^2}+25\right)\frac{x}{(\kappa-1)^{1/2}},
\end{align}
where we retain the definition $N_{\beta}=(\kappa-1)N$. Then, for these same $\beta$ values,
\begin{equation} \label{eq:de38}
\mathcal{D}_{\beta,N}^{(G)}\,\rho_{(1),\beta,N}^{(G)}(x)=0
\end{equation}
and
\begin{align} \label{eq:de39}
\mathcal{D}_{\beta,N}^{(G)}\,\frac{1}{N}W_{\beta,N}^{(G)}(x)&=\frac{2^{11}}{\sqrt{\kappa-1}}\left(\frac{4x^2}{\kappa-1}-6N_{\beta}-7\right)^2-3\frac{2^{12}}{\sqrt{\kappa-1}}.
\end{align}
\end{proposition}
\begin{proof}
Take equation \eqref{eq:de36} with $n=\beta=6$ and $N$ replaced by $N-1$ to obtain the matrix differential equation
\begin{equation} \label{eq:de40}
\frac{\mathrm{d}\mathbf{v}}{\mathrm{d}x}=\begin{bmatrix}2x&-4\sqrt{3}&0&0&0&0&0\\\frac{3N+5}{3\sqrt{3}}&\frac{4x}{3}&-\frac{10}{\sqrt{3}}&0&0&0&0\\0&\frac{6N+8}{3\sqrt{3}}&\frac{2x}{3}&-\frac{8}{\sqrt{3}}&0&0&0\\0&0&\sqrt{3}(N+1)&0&-2\sqrt{3}&0&0\\0&0&0&\frac{12N+8}{3\sqrt{3}}&-\frac{2x}{3}&-\frac{4}{\sqrt{3}}&0\\0&0&0&0&\frac{15N+5}{3\sqrt{3}}&-\frac{4x}{3}&-\frac{2}{\sqrt{3}}\\0&0&0&0&0&2\sqrt{3}N&-2x\end{bmatrix}\mathbf{v}
\end{equation}
for $\mathbf{v}=\left[\tilde{G}_{6,0}^{(N-1)}(x)\,\cdots\,\tilde{G}_{6,6}^{(N-1)}(x)\right]^T$. Like in the proofs of Lemmas \ref{L2.4} and \ref{L2.5}, for $1\leq p\leq6$, the $p\textsuperscript{th}$ row of the above matrix differential equation gives an expression for $\tilde{G}_{6,p}^{(N-1)}(x)$ in terms of $\frac{\mathrm{d}}{\mathrm{d}x}\tilde{G}_{6,p-1}^{(N-1)}(x)$ and $\tilde{G}_{6,k}^{(N-1)}(x)$ for $k<p$. Substituting these expressions into the equation corresponding to the last row in the order of decreasing $p$ then yields a seventh-order differential equation satisfied by $\tilde{G}_{6,0}^{(N-1)}(x)$. Since $\rho_{(1),6,N}^{(G)}(x)$ is proportional to $\tilde{G}_{6,0}^{(N-1)}(x)$, this equation is equivalent to \eqref{eq:de38} for $\beta=6$. Taking the Stieltjes transform of this result and substituting in the spectral moments $m_2^{(G)}$ and $m_4^{(G)}$ from \cite{WF14} then yields \eqref{eq:de39} for $\beta=6$. Employing the duality
\begin{equation} \label{eq:de41}
W_{\beta,N}^{(G)}(x)=\frac{\mathrm{i}}{\kappa\sqrt{\kappa}}W_{4/\beta,-\kappa N}^{(G)}(\mathrm{i}x/\sqrt{\kappa}),
\end{equation}
as is consistent with results of \cite{DE05,WF14}, then shows that \eqref{eq:de39} also holds for $\beta=2/3$. Finally, taking the inverse Stieltjes transform of this result shows that \eqref{eq:de38} holds for $\beta=2/3$ as well.
\end{proof}
Equation \eqref{eq:de38} has been checked for $N=1$ and $2$ using computer algebra. Similar to the Laguerre case, it seems like $x/\sqrt{\kappa-1}$ is the natural variable in Proposition \ref{P2.10}. This is evidently due to the duality \eqref{eq:de41} used in the proof of this proposition. Like in the Laguerre case, there is presently no benefit in changing variables to $x/\sqrt{\kappa-1}$.

It has been mentioned that equation \eqref{eq:de36} leads to a matrix differential equation which is equivalent to a scalar differential equation for $\rho_{(1),\beta,N+1}^{(G)}(x)$ when $\beta$ is even. For $\beta=2$ and $4$, these differential equations are respectively
\begin{equation} \label{eq:de42}
\frac{\mathrm{d}}{\mathrm{d}x}\begin{bmatrix}\tilde{G}_{2,0}^{(N-1)}(x)\\\tilde{G}_{2,1}^{(N-1)}(x)\\\tilde{G}_{2,2}^{(N-1)}(x)\end{bmatrix}=\begin{bmatrix}2x&-4&0\\N+1&0&-2\\0&2N&-2x\end{bmatrix}\begin{bmatrix}\tilde{G}_{2,0}^{(N-1)}(x)\\\tilde{G}_{2,1}^{(N-1)}(x)\\\tilde{G}_{2,2}^{(N-1)}(x)\end{bmatrix}
\end{equation}
and
\begin{equation} \label{eq:de43}
\frac{\mathrm{d}}{\mathrm{d}x}\begin{bmatrix}\tilde{G}_{4,0}^{(N-1)}(x)\\\tilde{G}_{4,1}^{(N-1)}(x)\\\tilde{G}_{4,2}^{(N-1)}(x)\\\tilde{G}_{4,3}^{(N-1)}(x)\\\tilde{G}_{4,4}^{(N-1)}(x)\end{bmatrix}=\begin{bmatrix}2x&-4\sqrt{2}&0&0&0\\\frac{2N+3}{2\sqrt{2}}&x&-3\sqrt{2}&0&0\\0&\sqrt{2}(N+1)&0&-2\sqrt{2}&0\\0&0&\frac{6N+3}{2\sqrt{2}}&-x&-\sqrt{2}\\0&0&0&2\sqrt{2}N&-2x\end{bmatrix}\begin{bmatrix}\tilde{G}_{4,0}^{(N-1)}(x)\\\tilde{G}_{4,1}^{(N-1)}(x)\\\tilde{G}_{4,2}^{(N-1)}(x)\\\tilde{G}_{4,3}^{(N-1)}(x)\\\tilde{G}_{4,4}^{(N-1)}(x)\end{bmatrix}.
\end{equation}
The corresponding scalar differential equations for the density agree with those of Proposition \ref{P2.1} after scaling $x\mapsto\sqrt{\frac{N\kappa}{2g}}x$. So too does the differential equation for $\rho_{(1),1,N}^{(G)}(x)$ obtained by applying duality \eqref{eq:de41}.

\setcounter{equation}{0}
\section{Recurrence Relations} \label{s3}
In this section, we present the first of our two promised applications of the differential equations derived in Section \ref{s2}. Namely, these differential equations yield recursions for the integer spectral moments $m_k$ \eqref{eq:intro6}. The negative-integer moments ($m_{-k}$ with $k>0$) have gained interest primarily due to their connection to problems of quantum transport \cite{BFB97,MS11,MS12,CMSV16b} and more recently due to duality principles relating the negative-integer moments $m_{-k}$ to their counterparts $m_{k-1}$ \cite{CMOS18}. On the other hand, the positive-integer moments have been studied more extensively and are known to be polynomials in $N$ in the Gaussian and Laguerre cases, and rational functions of $N$ in the Jacobi case. Thus, the positive-integer moments of the eigenvalue densities have expansions in $N$ (Gaussian and Laguerre ensembles) and $1/N$ (Jacobi ensembles) whose coefficients satisfy recursions presented in \S\ref{s3.2} which are in turn obtained from those in \S\ref{s3.1}. Recursions satisfied by these moment-expansion coefficients have been studied in special cases, with particular attention paid to their topological or combinatorial interpretations. For instance, the GUE moment coefficients satisfy the Harer-Zagier recursion \cite{HZ86}, and it was subsequently found that the Gaussian and Laguerre ensemble positive-integer moments all have interpretations in terms of ribbon graphs for each of $\beta=1,2$, and $4$ \cite{MW03,Di03,BP09}. Such ribbon graph interpretations are of interest in the field of enumerative geometry due to their relation to combinatorial maps, see e.g. \cite{LC09} and references therein.

In \S\ref{s3.3}, we use the results of Section \ref{s2} to obtain differential equations for the coefficients of the topological expansions \eqref{eq:intro8} of the scaled resolvents. These differential equations yield a recursive process for computing the resolvents up to any desired order in $1/N$. We remark that topological recursion accomplishes the same task, although less efficiently, for general $\beta>0$; see \cite{EO09,FRW17} and references therein. Hence, our work isolates extra structures for particular $\beta$. Checking that the differential equations of \S\ref{s3.3} are satisfied by the resolvent coefficients computed according to \cite{FRW17} will reaffirm the results of Section \ref{s2} and those of Section \ref{s3} pertaining to the positive-integer moments.

\subsection{Recursions for the spectral moments} \label{s3.1}
Obtaining moment recursions from the differential equations \eqref{eq:de19}, \eqref{eq:de24}, \eqref{eq:de29}, and \eqref{eq:de38} for $\rho_{(1),\beta,N}(x)$ is straightforward: Multiply both sides of the differential equations by $x^k$ and then integrate both sides. In the notation of Appendix \ref{A}, this is done by modifying the proof of say \eqref{eq:de20} and computing terms of the form $\mathcal{I}(0;p+k,q,n,0)$, i.e. replacing $(s-x)^{-k}$ by $x^k$ in \eqref{eq:A1}. The calculations follow through identically, owing to the fact that boundary terms vanish in all instances of integration by parts, presuming that the moments of interest converge. Thus, the Jacobi and Laguerre moment recursions presented in this subsection are valid for all positive-integer moments, but only extend to the negative integer moments $m_{-k}$ with $k<a+1$; only the positive-integer moments of the Gaussian ensembles converge. Moreover, the recurrences obtained by this method inter-relate general complex moments, when well defined.

An alternative proof for the positive-integer moments follows from utilising the geometric series in \eqref{eq:intro7} to see that
\begin{equation} \label{eq:rr1}
W_{\beta,N}(x)=\sum_{k=0}^{\infty}\frac{m_k}{x^{k+1}},\quad x\notin\textrm{supp}\,\rho_{(1),\beta,N}.
\end{equation}
Substituting this into the differential equations \eqref{eq:de20}, \eqref{eq:de25}, \eqref{eq:de30}, and \eqref{eq:de39} and then equating terms of equal order in $x$ gives relations between the moments. The equations obtained from terms of negative order in $x$ give the upcoming recursions on the positive-integer moments, while the terms of order one and positive order in $x$ give the first few moments required to run the recursions (the latter are also available in earlier literature; see e.g.~\cite{FRW17} and references therein). Both methods of proof apply to all of the cases considered, so the following propositions will be presented without formal verification.

It should be noted that the moment recursions pertaining to the Laguerre ensembles can be obtained from those for the Jacobi ensembles by substituting in $m_k^{(J)}=b^{-k}m_k^{(L)}$ and then taking the limit $b\rightarrow\infty$ (cf. equation \eqref{eq:de26}).

Applying our methods to the JUE density and resolvent differential equations recovers a recurrence of Cunden et al. \cite[Prop.~4.8]{CMOS18} wherein it is formulated as a recurrence for the differences of moments $\Delta m_k^{(J)}:=m_k^{(J)}-m_{k+1}^{(J)}$. This result is a refinement of a similar such recurrence of Ledoux's for the Jacobi ensemble shifted to have weight $w(x)=(x+1)^a(1-x)^b$ \cite{Le04}. We present the recurrence for completeness.

\begin{proposition} \label{P3.1}
For $\beta=2$ and $k\in\mathbb{Z}$, the convergent moments of the Jacobi ensemble satisfy the third-order linear recurrence
\begin{equation} \label{eq:rr2}
\sum_{l=0}^3 d_{2,l}^{(J)}m_{k-l}^{(J)}=0,
\end{equation}
where
\begin{align} \label{eq:rr3}
d_{2,0}^{(J)}&=k\left[(a+b+2N)^2-(k-1)^2\right], \nonumber
\\d_{2,1}^{(J)}&=3k^3-11k^2-k\left[2(a+b+2N)^2+a^2-b^2-14\right] \nonumber
\\&\quad+3(a+b)(a+2N)+6(N^2-1), \nonumber
\\d_{2,2}^{(J)}&=(2k-3)\left[2N(a+b+N)+ab\right]-(k-2)\left[3k^2-10k-3a^2+9\right], \nonumber
\\d_{2,3}^{(J)}&=(k-3)\left[(k-2)^2-a^2\right].
\end{align}
The initial terms $m_0^{(J)},m_1^{(J)},m_2^{(J)}$ are given in \cite{MRW15}.
\end{proposition}
The fact that this recurrence can be naturally formulated in terms of the differences of moments $\Delta m_k^{(J)}$ is equivalent to the observation that $\sum_{l=0}^3d_{2,l}^{(J)}=0$.  

The second-order recurrence on the LUE moments,
\begin{equation} \label{eq:rr4}
(k+1)m_k^{(L)}=(2k-1)(a+2N)m_{k-1}^{(L)}+(k-2)\left[(k-1)^2-a^2\right]m_{k-2}^{(L)},
\end{equation}
derived as a limiting case of \eqref{eq:rr2} according to the prescription in the paragraph preceding that which contains Proposition \ref{P3.1},
agrees with the recurrence obtained in \cite{Le04} for $\{m_k^{(L)}\}_{k\geq0}$, which was shown in \cite{CMSV16b} to hold for all $k>-a-1$.

\begin{proposition} \label{P3.2}
For $\beta=1$ and $4$, and $k\in\mathbb{Z}$, the convergent moments of the Jacobi ensemble satisfy the fifth-order linear recurrence
\begin{equation} \label{eq:rr5}
\sum_{l=0}^5 d_{4,l}^{(J)}m_{k-l}^{(J)}=0,
\end{equation}
where
\begin{align} \label{eq:rr6}
d_{4,0}^{(J)}&=k(\tilde{c}^2-(k-2)^2)(\tilde{c}^2-(2k-1)^2), \nonumber
\\d_{4,1}^{(J)}&=\frac{1}{2}(\tilde{c}^2-9)^2(5-6k)+\frac{1}{2}(\tilde{a}-\tilde{b})\left[(\tilde{c}^2-9)(5-4k)+2k(5(k-1)(k-5)+4k)\right] \nonumber
\\&\quad+(\tilde{c}^2-9)k\left[5(4k-3)(k-3)+2k\right]-4k^2(k-5)\left[5(k-2)(k-1)-2\right], \nonumber
\\d_{4,2}^{(J)}&=\tilde{c}^4 (3 k-5)+\tilde{c}^2 \left[\tfrac{1}{2}(\tilde{a}+\tilde{b})(2k-5)+5(\tilde{a}-\tilde{b})(k-2)\right] \nonumber
\\&\quad-\tilde{c}^2\left[30k^3-171k^2+339k-230\right]-\frac{1}{2} (\tilde{a}+\tilde{b}) \left[5 k^3-44 k^2+129 k-125\right] \nonumber
\\&\quad-\frac{1}{2} (\tilde{a}-\tilde{b}) \left[35 k^3-246 k^2+581 k-460\right]+\frac{1}{2} (2 k-5) (\tilde{a}-\tilde{b})^2 \nonumber
\\&\quad+40k^4(k-11)+1966k^3-4443k^2+5056k-2305, \nonumber
\\d_{4,3}^{(J)}&=\frac{1}{2} \tilde{c}^4 (5-2 k)+\tilde{c}^2 \left[\tfrac{1}{4}(\tilde{a}+\tilde{b})(25-8 k)+\tfrac{1}{4}(\tilde{a}-\tilde{b})(45-16k)\right] \nonumber
\\&\quad+\tilde{c}^2\left[20 k^3-155 k^2+401 k-345\right]+\frac{5}{4}(\tilde{a}+\tilde{b})\left[6k^3-62k^2+216k-253\right] \nonumber
\\&\quad+\frac{1}{4} (\tilde{a}-\tilde{b})\left[90 k^3-806 k^2+2436 k-2485\right]+\frac{1}{4}(\tilde{a}^2-\tilde{b}^2)(15-4 k) \nonumber
\\&\quad+\frac{1}{4}(\tilde{a}-\tilde{b})^2(25-8 k)-4k^3(10k^2-140k+789) +8923 k^2-12600 k+\tfrac{14125}{2}, \nonumber
\\d_{4,4}^{(J)}&=(k-4) \left[k^3+k^2-18 k-\left(\tilde{c}^2-\tilde{b}-4 k^2+29 k-51\right)\left(5 k^2-29 k+40\right)\right] \nonumber
\\&\quad+\frac{1}{2}\tilde{a}^2(6 k-25)+\frac{1}{2}\tilde{a}\left[(4 k-15) \left(\tilde{c}^2-\tilde{b}-10 k^2+76k-147\right)-2 (k-5)\right], \nonumber
\\d_{4,5}^{(J)}&=(k-5)\left[4(k-5)(k-4)-\tilde{a}\right]\left[\tilde{a}-(k-4)(k-2)\right],
\end{align}
and we retain the definitions of $\tilde{a},\tilde{b}$, and $\tilde{c}$ given in Theorem \ref{T2.8}. The required initial terms $m_0^{(J)}$ to $m_4^{(J)}$ can be computed through MOPS \cite{MOPS} or via the methods presented in \cite{MRW15,FRW17}.
\end{proposition}
We again see that $\sum_{l=0}^5d_{4,l}^{(J)}=0$, suggesting that it is natural to rewrite this recurrence in terms of the $\Delta m_k^{(J)}$. Moving on, applying the aforementioned limiting procedure to Proposition \ref{P3.2} yields the analogous recurrences in the Laguerre case.

\begin{proposition} \label{P3.3}
For $\beta=1$ and $4$, and $k\in\mathbb{Z}$, the convergent moments of the Laguerre ensemble satisfy the fourth-order linear recurrence
\begin{equation} \label{eq:rr7}
\sum_{l=0}^4 d_{4,l}^{(L)}(\kappa-1)^lm_{k-l}^{(L)}=0,
\end{equation}
where
\begin{align} \label{eq:rr8}
d_{4,0}^{(L)}&=k+1, \nonumber
\\d_{4,1}^{(L)}&=(1-4k)(a_{\beta}+4N_{\beta}), \nonumber
\\d_{4,2}^{(L)}&=(1-k)(5k^2-11k+4)+(2k-3)\left[\tilde{a}+2(a_{\beta}+4N_{\beta})^2\right], \nonumber
\\d_{4,3}^{(L)}&=(a_{\beta}+4N_{\beta})\left[(11-4k)\tilde{a}+10k^3-68k^2+146k-96\right], \nonumber
\\d_{4,4}^{(L)}&=(k-4)\left[\tilde{a}-4(k-4)(k-3)\right]\left[\tilde{a}-(k-3)(k-1)\right],
\end{align}
and we retain the definitions of $a_{\beta},\tilde{a}$, and $N_{\beta}$ given in Theorem \ref{T2.8}. The required initial terms $m_0^{(L)}$ to $m_3^{(L)}$ are given in \cite{MRW15,FRW17}.
\end{proposition}
This recursion is a homogeneous version of \cite[Eq.~(43)]{CMSV16b}, in which the inhomogeneous terms depend on the moments of the LUE. Another simple observation is that for each of $\beta=1,2$, and $4$, the recursions for the spectral moments of the Laguerre ensembles contain one less term than the corresponding recursions for the Jacobi ensembles, and these terms are drastically simpler, as well. This is in keeping with the loss of parameter $b$.

We remark that our methods applied to Proposition \ref{P2.1} yield moment recursions for the GOE, GUE, and GSE which agree with those given in \cite{HZ86,Le09,WF14},
and moreover, agree with the appropriate limit of the corresponding Laguerre or Jacobi recurrences as derived above. Our methods allow us to go further and derive analogous recursions for the $\beta=2/3$ and $\beta=6$ Gaussian ensembles' spectral moments.

\begin{proposition} \label{P3.4}
For $\beta=2/3$ and $6$, and $k\geq12$, the moments of the Gaussian ensemble satisfy the sixth-order linear recurrence
\begin{equation} \label{eq:rr9}
\sum_{l=0}^6d_{6,l}^{(G)}\left(\frac{\kappa-1}{4}\right)^lm_{k-2l}^{(G)},
\end{equation}
where
\begin{align} \label{eq:rr10}
d_{6,0}^{(G)}&=-4(k+2), \nonumber
\\d_{6,1}^{(G)}&=8(3k-1)(3N_{\beta}+2), \nonumber
\\d_{6,2}^{(G)}&=48(8-3k)N_{\beta}(3N_{\beta}+4)+49k^3-216k^2+92k+320, \nonumber
\\d_{6,3}^{(G)}&=4(k-5)(3N_{\beta}+2)\left[24N_{\beta}(3N_{\beta}+4)-49k^2+304k-442\right], \nonumber
\\d_{6,4}^{(G)}&= 2(k-5)_3\left[294N_{\beta}(3N_{\beta}+4)-63k(k-6)-274\right], \nonumber
\\d_{6,5}^{(G)}&=252(k-5)_5(3N_{\beta}+2), \nonumber
\\d_{6,6}^{(G)}&=81(k-5)_7,
\end{align}
and we retain the definition $N_{\beta}=(\kappa-1)N$. Here, $(x)_n=x(x-1)\cdots(x-n+1)$ is the falling Pochhammer symbol. The required initial terms $m_0^{(G)}, m_2^{(G)},\ldots,m_{10}^{(G)}$ are given in \cite{WF14}.
\end{proposition}
This recurrence is different to those given earlier in this subsection, in that it runs over every second moment. This is because the odd moments vanish, since the eigenvalue density $\rho_{(1),\beta,N}^{(G)}(x)$ of the Gaussian ensemble is an even function. Indeed, equation \eqref{eq:rr9} holds trivially for $k$ odd. The methods employed in \cite{Le09} manifest a coupling between the moments of the GUE and GOE, which remains mysterious from our viewpoint, as does the coupled recurrence between the LUE and LOE moments given in \cite[Eq.~(43)]{CMSV16b}: as made clear in \cite{CMSV16b}, both can be traced back to a structural formula for the $\beta=1$ density in the classical cases, in terms of the $\beta=2$ density plus what can be regarded as a rank $1$ correction; see \cite{AFNV00}. There is no evidence that such a coupling exists between the moments of the $\beta=6$ Gaussian ensemble and the moments of the GOE and/or GUE.

\begin{remark}
It was shown in \cite{CMOS18} that the LUE moments and the JUE moment differences satisfy the following reciprocity laws:
\begin{align*}
m_{-k-1}^{(L)}&=\left(\prod_{j=-k}^{j=k}\frac{1}{a-j}\right)m_k^{(L)},
\\\Delta m_{-k-1}^{(J)}&=\left(\prod_{j=-k}^{j=k}\frac{a+b+2N-j}{a-j}\right)\Delta m_k^{(J)}.
\end{align*}
It is not immediately obvious that there exist similar laws for the moments of the orthogonal or symplectic ensembles; one can experiment using data computable from Propositions \ref{P3.2} and \ref{P3.3}.
\end{remark}

\subsection{Harer-Zagier type $1$-point recursions} \label{s3.2}
It is known from Jack polynomial theory \cite{DE05,DP12,MRW15} that the spectral moments for the Gaussian, Laguerre, and Jacobi ensembles have, for $k>0$, the structures
\begin{align}
m_{2k}^{(G)}&=\sum_{l=0}^{k}M_{k,l}^{(G)}N^{k-l+1}, \label{eq:rr11}
\\m_k^{(L)}&=\sum_{l=0}^{k}M_{k,l}^{(L)}N^{k-l+1}, \label{eq:rr12}
\\m_k^{(J)}&=\sum_{l=0}^{\infty}M_{k,l}^{(J)}N^{1-l}, \label{eq:rr13}
\end{align}
respectively. Recall that the odd moments of the Gaussian ensemble are zero, and that in general the positive-integer moments of the Jacobi ensemble are rational in $N$. Here, the expansion coefficients $M_{k,l}^{(\,\cdot\,)}$ have no dependence on $N$. Substituting these expressions into the recurrences given in the previous subsection and then equating terms of equal order in $N$ yield so-called $1$-point recursions \cite{CD18}.

Applying this procedure to the GUE retrieves the celebrated Harer-Zagier recursion \cite{HZ86}, while the GOE analogue has been studied in \cite{Le09} and references therein. The motivation for the former came from the interpretation of $M_{k,l}^{(GUE)}$ as the number of ways of gluing the edges of a $2k$-gon to form a compact orientable surface of genus $l/2$. Similarly, $4^kM_{k,l}^{(GOE)}$ counts the number of gluings of a $2k$-gon that form a compact locally-orientable surface of Euler characteristic $2-l$ \cite{KK03}. It is rather straightforward to obtain the GSE analogue through the $\beta\leftrightarrow4/\beta$ duality. Proposition \ref{P3.4} allows an extension of the GUE, GOE, GSE $1$-point recursions to the $\beta=2/3$ and $\beta=6$ Gaussian ensembles. Missing are the initial conditions. For this and the other recurrences presented below, these can be computed according to the strategy outlined in the second paragraph of \S\ref{s3.1}.

\begin{proposition} \label{P3.5}
Expand the moments of the Gaussian ensemble according to \eqref{eq:rr11}. Then, for $\beta=2/3$ and $6$, and $k\geq6$,
\begin{equation} \label{eq:rr14}
(k+1)M_{k,l}^{(G)}=\sum_{i=1}^6\sum_{j=0}^i\frac{(\kappa-1)^{2i-j}}{2^i}f_{i,j}M_{k-i,l-j}^{(G)},
\end{equation}
where
\begin{equation*}
\begin{gathered}
f_{1,0}=3(6k-1),\quad f_{1,1}=2(6k-1),\quad f_{2,0}=36(4-3k),
\\f_{2,1}=48(4-3k),\quad f_{2,2}=49k^3-108k^2+23k+40,\quad f_{3,0}=108(2k-5),
\\f_{3,1}=216(2k-5),\quad f_{3,2}=3(5-2k)(98k^2-304k+189),
\end{gathered}
\end{equation*}
\begin{equation} \label{eq:rr15}
\begin{gathered}
f_{3,3}=2(5-2k)(98k^2-304k+221),\quad f_{4,2}=\tfrac{441}{2}(2k-5)_3,\quad f_{4,3}=294(2k-5)_3,
\\f_{4,4}=\tfrac{1}{2}(2k-5)_3\left(126k(3-k)-137\right),\quad f_{5,4}=\tfrac{189}{2}(2k-5)_5,
\\f_{5,5}=63(2k-5)_5,\quad f_{6,6}=\tfrac{81}{8}(2k-5)_7,
\end{gathered}
\end{equation}
and all other $f_{i,j}$ are zero. We also set $M_{k,l}^{(G)}=0$ if $l<0$ or $l>k$.
\end{proposition}
When \eqref{eq:rr14} is used to compute $M_{k,0}^{(G)}$, it reduces to
\begin{multline} \label{eq:rr16}
16(k+1)M_{k,0}^{(G)}=12(6k-1)(\kappa-1)^2M_{k-1,0}^{(G)}-36(3k-4)(\kappa-1)^4M_{k-2,0}^{(G)}
\\+27(2k-5)(\kappa-1)^6M_{k-3,0}^{(G)},\quad\kappa=1/3\textrm{ or }3.
\end{multline}
This is in keeping with the limiting scaled density of the Gaussian ensembles equaling the Wigner semi-circle law specified by the density $\frac{1}{\pi}\sqrt{2-x^2}$ supported on $|x|<\sqrt{2}$: up to a scale factor the Catalan numbers are the even moments.

To expand the Laguerre and Jacobi spectral moments in $N$, we need to decide on how the parameters $a$ and $b$ scale with $N$, which has been inconsequential up until now. To make our results suitable for most known (see e.g. \cite{MS11}) or potential applications, we first consider the LUE moments with $a=\alpha_1N+\delta_1$ where the $\alpha_1$ and $\delta_1$ are parameters of order unity, though our method easily accommodates for more general $N$-expansions of $a$. Substituting this together with the expansion \eqref{eq:rr12} into \eqref{eq:rr4} gives the sought recurrence.

\begin{proposition} \label{P3.6}
For $\beta=2$, and $k\geq2$,
\begin{multline} \label{eq:rr17}
(k+1)M_{k,l}^{(L)}=(2k-1)((2+\alpha_1)M_{k-1,l}^{(L)}+\delta_1M_{k-1,l-1}^{(L)})
\\-\alpha_1^2(k-2)M_{k-2,l}^{(L)}+(k-2)\left((k-1)^2-\delta_1^2\right)M_{k-2,l-2}^{(L)},
\end{multline}
where we set $M_{k,l}^{(L)}=0$ if $l<0$ or $l>k$.
\end{proposition}
When this is used to compute $M_{k,0}^{(L)}$ in the case $a=\delta_1={\rm O}(1)$ and thus $\alpha_1=0$, one retrieves the familiar Catalan recursion
\begin{equation*}
M_{k,0}^{(L)}=\frac{2(2k-1)}{k+1}M_{k-1,0}^{(L)}.
\end{equation*}
This is consistent with the fact that for $a={\rm O}(1)$, the limiting scaled spectral density for the Laguerre $\beta$ ensemble is given by the particular Marchenko-Pastur density $\frac{1}{2\pi}\sqrt{4/x-1}$ \cite{PS11}, which has for its $k\textsuperscript{th}$ moment the $k\textsuperscript{th}$ Catalan number.

When $\delta_1=0$, equation \eqref{eq:rr17} reduces to a recursion over even $l$, and we may interpret $l/2$ as genus following Di Francesco \cite{Di03}. In this interpretation, $M_{k,l}^{(L)}$ counts the same gluings as those in the Harer-Zagier recursion, except that the base $2k$-gon has vertices alternately coloured black and white and the gluings respect this bicolouring. Moreover, the coloured vertices are weighted according to the dimensions of the underlying complex random matrices, which relate to the $a$ parameter in our notation. When we additionally have $\alpha_1=0$ so that $a=0$, the recurrence \eqref{eq:rr17} is equivalent to that of \cite[Theorem~4.1]{ND18} on the aforementioned gluings with the coloured vertices no longer weighted.

When we take $\delta_1\neq0$, the recurrence \eqref{eq:rr17} is now over all integer $l\geq0$, to which the interpretation in terms of gluings of $2k$-gons does not immediately extend. On the other hand, most of the literature \cite{Be97,BFB97,CMSV16a,CMSV16b,VV08,No08,LV11,MS11,MS12} relating the Laguerre and Jacobi ensembles to the problems of quantum cavities and quantum transport either does not specify the scaling behaviour of the $a$ and $b$ parameters, fixes $a$ and $b$ as precise values, or takes $a,b\propto N$. In light of this, we proceed with $a=\alpha_1N$ and $b=\alpha_2N$ with $\alpha_i={\rm O}(1)$ for the sake of brevity and clarity. The derivation is readily transferable to more general cases should such a need arise in the future. To complete this subsection, we present the $M_{k,l}^{(\,\cdot\,)}$ recurrences for the LOE, LSE, and JUE. We do not present the JOE and JSE $M_{k,l}^{(J)}$ recurrences due to their (relatively speaking) unwieldy form, but note that they can also be derived from the results of \S\ref{s3.1} in the same way as the other recurrences in this subsection.

\begin{proposition} \label{P3.7}
For $\beta=1$ and $4$, $a=\alpha_1N$ with $\alpha_1={\rm O}(1)$, and $k\geq4$,
\begin{equation} \label{eq:rr18}
(k+1)M_{k,l}^{(L)}=\sum_{i=1}^4\sum_{j=0}^i(\kappa-1)^jg_{i,j}M_{k-i,l-j}^{(L)},
\end{equation}
where
\begin{equation} \label{eq:rr19}
\begin{gathered}
g_{1,0}=(4k-1)\left(\alpha_1+4(\kappa-1)\right)(\kappa-1),
\\g_{2,0}=(3-2k)\left(\alpha_1^2+2(\alpha_1+4(\kappa-1))^2(\kappa-1)^2\right),\quad g_{2,1}=2(2k-3)\alpha_1,
\\g_{2,2}=(k-1)\left(5k^2-11k+4\right),\quad g_{3,0}=(4k-11)(\alpha_1+4(\kappa-1))\alpha_1^2(\kappa-1),
\\g_{3,1}=2(11-4k)(\alpha_1+4(\kappa-1))\alpha_1(\kappa-1),
\\g_{3,2}=2(3-k)\left(5k^2-19k+16\right)(\alpha_1+4(\kappa-1))(\kappa-1),\quad g_{4,0}=(4-k)\alpha_1^4,
\\g_{4,1}=4(k-4)\alpha_1^3,\quad g_{4,2}=(k-4)\left(5k^2-32k+47\right)\alpha_1^2,
\\g_{4,3}=2(4-k)(k-3)(5k-17)\alpha_1,\quad g_{4,4}=4(1-k)\left[(k-4)(k-3)\right]^2,
\end{gathered}
\end{equation}
and all other $g_{i,j}$ are zero. We also set $M_{k,l}^{(L)}=0$ if $l<0$ or $l>k$.
\end{proposition}

\begin{proposition} \label{P3.8}
For $\beta=2$, $a=\alpha_1N$ and $b=\alpha_2N$, with $\alpha_i={\rm O}(1)$, and $k\geq3$,
\begin{equation} \label{eq:rr20}
k(\alpha_1+\alpha_2+2)^2M_{k,l}^{(J)}=k(k-1)^2M_{k,l-2}^{(J)}+\sum_{i=1}^3\left(h_{i,0}M_{k-i,l}^{(J)}+h_{i,1}M_{k-i,l-2}^{(J)}\right),
\end{equation}
where
\begin{equation} \label{eq:rr21}
\begin{gathered}
h_{1,0}=2(4k-3)(\alpha_1+\alpha_2+1)+\left(3\alpha_1(k-1)+\alpha_2k\right)(\alpha_1+\alpha_2),
\\h_{1,1}=(1-k)(3k^2-8k+6),
\\h_{2,0}=3\alpha_1^2(2-k)+(3-2k)\left((\alpha_1+2)(\alpha_2+2)-2\right),
\\ h_{2,1}=(k-2)(3k^2-10k+9),\quad h_{3,0}=\alpha_1^2(k-3),\quad h_{3,1}=(3-k)(k-2)^2,
\end{gathered}
\end{equation}
and we set $M_{k,l}^{(J)}=0$ if $l<0$.
\end{proposition}
\begin{proof}
Since the derivation of this recurrence is slightly different to that described earlier, we supply the details. Substitute expansion \eqref{eq:rr13} into \eqref{eq:rr2}, along with the substitutions $a=\alpha_1N$ and $b=\alpha_2N$, and then equate terms of order $3-l$ in $N$. Note that with our choice of $a$ and $b$, the coefficients $d_{2,0}^{(J)},\ldots,d_{2,3}^{(J)}$ all contain a term of order two and a term of order one in $N$, and no other terms. The ${\rm O}(N^2)$ terms of $d_{2,i}^{(J)}$ correspond to the coefficients of $M_{k-i,l}^{(J)}$ in \eqref{eq:rr20} while the terms of order one correspond to the coefficients of $M_{k-i,l-2}^{(J)}$. 
\end{proof}

There are two immediate observations relating to Proposition \ref{P3.8}: Firstly, this recursion runs over even $l$, similar to Proposition \ref{P3.6} with $a=\alpha_1N$. Indeed, when the first few moments $m_0^{(J)},m_1^{(J)},m_2^{(J)}$ are expanded as series in $1/N$, they do not contain terms of even powers in $N$, so recurrence \eqref{eq:rr20} is trivially satisfied when $l$ is odd. Secondly, unlike the analogous recurrence relations for the GUE and LUE, $M_{k,l}^{(J)}$ depends via this recurrence on $M_{k,l-2}^{(J)}$ instead of just $M_{i,j}^{(J)}$ with $i<k$. This means that computing the order $N^l$ term of $m_k^{(J)}$ through this recurrence requires one to find all other terms in $m_k^{(J)}$ that are higher order in $N$.

As yet, it is not known whether the JUE spectral moments count some type of surface similar to those counted by the GUE and LUE spectral moments. If such an interpretation is possible for the $M_{k,l}^{(J)}$, the fact that recurrence \eqref{eq:rr20} runs over even $l$ suggests that $l/2$ might again play the role of genus, and the second observation above might give a clue as to what sets the JUE apart from the GUE and LUE. Since setting $a=0$ in Proposition \ref{P3.6} yields a simpler recurrence that retains an interpretation in terms of gluings of $2k$-gons, we set $a=b=0$ in Proposition \ref{P3.8} for comparison:

\begin{proposition} \label{P3.9}
In the context of Proposition \ref{P3.8} with $a=b=0$,
\begin{multline} \label{eq:rr22}
4kM_{k,l}^{(J)}=k(k-1)^2M_{k,l-2}^{(J)}+2(4k-3)M_{k-1,l}^{(J)}
\\+(1-k)(3k^2-8k+6)M_{k-1,l-2}^{(J)}+2(3-2k)M_{k-2,l}^{(J)}
\\+(k-2)(3k^2-10k+9)M_{k-2,l-2}^{(J)}+(3-k)(k-2)^2M_{k-3,l-2}^{(J)}.
\end{multline}
\end{proposition}

\subsection{Differential equations for the coefficients of the topological expansion} \label{s3.3}
A feature of the moment expansions \eqref{eq:rr11} and \eqref{eq:rr12} for the Gaussian and Laguerre ensembles is that for large $N$ they are proportional to $N^{k+1}$, whereas the moment expansion \eqref{eq:rr13} is proportional to $N$. For the corresponding resolvents to admit a $1/N$ expansion of the form given in \eqref{eq:intro8}, a change of scale and of normalisation is required, so that all moments are to leading order unity. For this, the rescaled spectral density will be normalised to integrate to unity and will have compact support in the $N\rightarrow\infty$ regime.

Following \cite{WF14} and \cite{FRW17}, this can be achieved by introducing the so-called smoothed (also referred to as global scaled) densities as
\begin{align}
\tilde{\rho}_{(1),\beta,N}^{(G)}(x)&=\sqrt{\tfrac{\kappa}{N}}\,\rho_{(1),\beta,N}^{(G)}(\sqrt{N\kappa}x), \label{eq:rr23}
\\\tilde{\rho}_{(1),\beta,N}^{(L)}(x)&=\kappa\,\rho_{(1),\beta,N}^{(L)}(N\kappa x), \label{eq:rr24}
\\\tilde{\rho}_{(1),\beta,N}^{(J)}(x)&=\tfrac{1}{N}\,\rho_{(1),\beta,N}^{(J)}(x), \label{eq:rr25}
\end{align}
and their corresponding scaled resolvents as
\begin{align}
\tilde{W}_{\beta,N}^{(G)}(x)&=\sqrt{N\kappa}\,W_{\beta,N}^{(G)}(\sqrt{N\kappa}x), \label{eq:rr26}
\\\tilde{W}_{\beta,N}^{(L)}(x)&=N\kappa\,W_{\beta,N}^{(L)}(N\kappa x); \label{eq:rr27}
\end{align}
since $\tilde{\rho}_{(1),\beta,N}^{(J)}(x)$ has the same support as $\rho_{(1),\beta,N}^{(J)}(x)$, there is no need to scale $W_{\beta,N}^{(J)}(x)$. In the $N\to\infty$ limit, the smoothed densities approach the Wigner semi-circle law in the Gaussian case, the Marchenko-Pastur law in the Laguerre case, and a functional form first deduced by Wachter in the Jacobi case. The exact functional forms of the latter two are dependent on the proportionality in $N$ of the $a$ and $b$ parameters, and are given explicitly in \cite{FRW17}. However, as with the Wigner semi-circle law, they are independent of $\kappa$.

Compared to the resolvents $W_{\beta,N}(x)$ of the eigenvalue densities, the scaled resolvents are generating functions for the moments of the smoothed densities above, which converge for large $x$, and can be expanded in $1/N$ according to
\begin{align}
\tilde{W}_{\beta,N}^{(G)}(x)&=2N\sum_{l=0}^{\infty}\frac{W_{\beta}^{(G),l}(x)}{(2N\sqrt{\kappa})^l}, \label{eq:rr28}
\\\tilde{W}_{\beta,N}^{(L)}(x)&=N\sum_{l=0}^{\infty}\frac{W_{\beta}^{(L),l}(x)}{(N\sqrt{\kappa})^l}, \label{eq:rr29}
\\W_{\beta,N}^{(J)}(x)&=N\sum_{l=0}^{\infty}\frac{W_{\beta}^{(J),l}(x)}{(N\kappa)^l}, \label{eq:rr30}
\end{align}
where we again follow \cite{WF14,FRW17} for consistency. Note that the smoothed densities \eqref{eq:rr23}--\eqref{eq:rr25} do not permit such $1/N$ expansions, due to oscillatory terms beyond leading order; see e.g.~\cite{GFF05}. Appropriately scaling the differential equations of Section \ref{s2} yields differential equations for the scaled resolvents which in turn give first-order differential equations for the expansion coefficients $W_{\beta}^l(x)$, after substituting in the expansions above and then equating terms of equal order in $N$. Compared to the topological recursion, these differential equations provide a tractable recursive process for computing the $W_{\beta}^l(x)$, which is simpler in that there is no need for introducing multi-point correlators, but more restrictive in that they only apply for the $\beta$-values considered in this paper. Presently, we intend to use the upcoming differential equations to check earlier results of this section for consistency both internally and with \cite{WF14,FRW17}. For this reason, we continue to treat the case $a=\alpha_1N$ and $b=\alpha_2N$, with $\alpha_i={\rm O}(1)$.

We begin with the Gaussian ensemble with $\beta\in\{2/3,1,4,6\}$ and refer to \cite{HT12} for the $\beta=2$ case.
\begin{proposition} \label{P3.10}
We reuse the notation of Proposition \ref{P2.1} with $g=1/2$ so that $h=\sqrt{\kappa}-1/\sqrt{\kappa}$ and $y_{(G)}=\sqrt{x^2-2}$. Then for $\beta=1$ and $4$, the expansion coefficients of $\tilde{W}_{\beta,N}^{(G)}(x)$ \eqref{eq:rr28} satisfy the differential equations
\begin{equation} \label{eq:rr31}
y_{(G)}^2\frac{\mathrm{d}}{\mathrm{d}x}W_{\beta}^{(G),0}(x)-xW_{\beta}^{(G),0}(x)=-1,
\end{equation}

\begin{equation} \label{eq:rr32}
y_{(G)}^2\frac{\mathrm{d}}{\mathrm{d}x}W_{\beta}^{(G),1}(x)-xW_{\beta}^{(G),1}(x)=4h\frac{\mathrm{d}}{\mathrm{d}x}W_{\beta}^{(G),0}(x)-\frac{h}{y_{(G)}^2}\left[2xW_{\beta}^{(G),0}(x)+5\right],
\end{equation}
and for $l\geq2$, the general differential equation

\begin{multline} \label{eq:rr33}
y_{(G)}^2\frac{\mathrm{d}}{\mathrm{d}x}W_{\beta}^{(G),l}(x)-xW_{\beta}^{(G),l}(x)=4h\frac{\mathrm{d}}{\mathrm{d}x}W_{\beta}^{(G),l-1}(x)-\frac{2hx}{y_{(G)}^2}W_{\beta}^{(G),l-1}(x)
\\+\frac{1}{y_{(G)}^2}\left[\frac{5y_{(G)}^2}{2}\frac{\mathrm{d}^3}{\mathrm{d}x^3}-3x\frac{\mathrm{d}^2}{\mathrm{d}x^2}+\frac{\mathrm{d}}{\mathrm{d}x}\right]W_{\beta}^{(G),l-2}(x)
\\-\frac{5h}{y_{(G)}^2}\frac{\mathrm{d}^3}{\mathrm{d}x^3}W_{\beta}^{(G),l-3}(x)+\frac{1}{y_{(G)}^2}\frac{\mathrm{d}^5}{\mathrm{d}x^5}W_{\beta}^{(G),l-4}(x),
\end{multline}
where we set $W_{\beta}^{(G),k}:=0$ for $k<0$.
\end{proposition}

\begin{proposition} \label{P3.11}
Retain the choice of $g=1/2$ from Proposition \ref{P3.10}. Then for $\beta=2/3$ and $6$, the expansion coefficients of $\tilde{W}_{\beta,N}^{(G)}(x)$ satisfy the differential equations
\begin{equation} \label{eq:rr34}
y_{(G)}^2\frac{\mathrm{d}}{\mathrm{d}x}W_{\beta}^{(G),0}(x)-xW_{\beta}^{(G),0}(x)=-1,
\end{equation}

\begin{equation} \label{eq:rr35}
y_{(G)}^2\frac{\mathrm{d}}{\mathrm{d}x}W_{\beta}^{(G),1}(x)-xW_{\beta}^{(G),1}(x)=6h\frac{\mathrm{d}}{\mathrm{d}x}W_{\beta}^{(G),0}(x)-\frac{h}{y_{(G)}^2}\left[4xW_{\beta}^{(G),0}(x)-7\right],
\end{equation}

\begin{multline} \label{eq:rr36}
y_{(G)}^2\frac{\mathrm{d}}{\mathrm{d}x}W_{\beta}^{(G),2}(x)-xW_{\beta}^{(G),2}(x)=6h\frac{\mathrm{d}}{\mathrm{d}x}W_{\beta}^{(G),1}(x)-\frac{4hx}{y_{(G)}^2}W_{\beta}^{(G),1}(x)-\frac{43}{3y_{(G)}^4}
\\+\frac{1}{12y_{(G)}^4}\left[49y_{(G)}^4\frac{\mathrm{d}^3}{\mathrm{d}x^3}-78xy_{(G)}^2\frac{\mathrm{d}^2}{\mathrm{d}x^2}+3(72-19x^2)\frac{\mathrm{d}}{\mathrm{d}x}+25x\right]W_{\beta}^{(G),0}(x),
\end{multline}
and for $l\geq3$, the general differential equation
\begin{multline} \label{eq:rr37}
y_{(G)}^2\frac{\mathrm{d}}{\mathrm{d}x}W_{\beta}^{(G),l}(x)-xW_{\beta}^{(G),l}(x)=6h\frac{\mathrm{d}}{\mathrm{d}x}W_{\beta}^{(G),l-1}(x)-\frac{4hx}{y_{(G)}^2}W_{\beta}^{(G),l-1}(x)
\\+\frac{1}{12y_{(G)}^4}\left[49y_{(G)}^4\frac{\mathrm{d}^3}{\mathrm{d}x^3}-78xy_{(G)}^2\frac{\mathrm{d}^2}{\mathrm{d}x^2}+3(72-19x^2)\frac{\mathrm{d}}{\mathrm{d}x}+25x\right]W_{\beta}^{(G),l-2}(x)
\\+\frac{h}{3y_{(G)}^4}\left[49y_{(G)}^2\frac{\mathrm{d}^3}{\mathrm{d}x^3}+39x\frac{\mathrm{d}^2}{\mathrm{d}x^2}-10\frac{\mathrm{d}}{\mathrm{d}x}\right]W_{\beta}^{(G),l-3}+\frac{7h}{y_{(G)}^4}\frac{\mathrm{d}^5}{\mathrm{d}x^5}W_{\beta}^{(G),l-5}
\\+\frac{1}{18y_{(G)}^4}\left[63y_{(G)}^2\frac{\mathrm{d}^5}{\mathrm{d}x^5}+63x\frac{\mathrm{d}^4}{\mathrm{d}x^4}+230\frac{\mathrm{d}^3}{\mathrm{d}x^3}\right]W_{\beta}^{(G),l-4}+\frac{3h}{4y_{(G)}^4}\frac{\mathrm{d}^7}{\mathrm{d}x^7}W_{\beta}^{(G),l-6},
\end{multline}
where we set $W_{\beta}^{(G),k}:=0$ for $k<0$.
\end{proposition}
This proposition has been checked against \cite{WF14} up to $l=6$, and thus also serves as a check for differential equation \eqref{eq:de39} up to order six in $1/N$. Interpreting $W_{\beta}^{(G),l}(x)$ as a generating function for $M_{k,l}^{(G)}$, we confirm Proposition \ref{P3.5} up to $l=6$ as well. Similar checks for consistency have been carried out for Propositions \ref{P3.12} to \ref{P3.15} below.

Moving on to the Laguerre ensembles, we have the following three propositions.
\begin{proposition} \label{P3.12}
The expansion coefficients of the LUE scaled density's resolvent $\tilde{W}_{2,N}^{(L)}$ \eqref{eq:rr29} satisfy the differential equation
\begin{equation} \label{eq:rr38}
\left[4x-(\alpha_1-x)^2\right]x\frac{\mathrm{d}}{\mathrm{d}x}W_2^{(L),0}(x)+\left[(\alpha_1+2)x-\alpha_1^2\right]W_2^{(L),0}(x)=x+\alpha_1,
\end{equation}
and for $l\geq2$, the general differential equation
\begin{multline} \label{eq:rr39}
\left[(\alpha_1-x)^2-4x\right]x\frac{\mathrm{d}}{\mathrm{d}x}W_2^{(L),l}(x)-\left[(\alpha_1+2)x-\alpha_1^2\right]W_2^{(L),l}(x)
\\=\left[x^3\frac{\mathrm{d}^3}{\mathrm{d}x^3}+4x^2\frac{\mathrm{d}^2}{\mathrm{d}x^2}+2x\frac{\mathrm{d}}{\mathrm{d}x}\right]W_2^{(L),l-2}(x).
\end{multline}
\end{proposition}

In the $a=\alpha_1N$ setting considered here, $W_2^{(L),k}=0$ when $k$ is odd. Thus, differential-difference equation \eqref{eq:rr39} holds vacuously for odd $l$, and should otherwise be interpreted as a recursion over even $l$. This is in keeping with the discussion following Proposition \ref{P3.6}.

\begin{proposition} \label{P3.13}
Let $y_{(L),1}=\sqrt{(2\alpha_1-x)^2-4x}$. The expansion coefficients of the LOE scaled density's resolvent $\tilde{W}_{1,N}^{(L)}(x)$ satisfy the differential equations
\begin{equation} \label{eq:rr40}
-xy_{(L),1}^2\frac{\mathrm{d}}{\mathrm{d}x}W_1^{(L),0}(x)+2\left[(\alpha_1+1)x-2\alpha_1^2\right]W_1^{(L),0}(x)=x+2\alpha_1,
\end{equation}

\begin{multline} \label{eq:rr41}
xy_{(L),1}^4\frac{\mathrm{d}}{\mathrm{d}x}W_1^{(L),1}(x)-2y_{(L),1}^2\left[(\alpha_1+1)x-2\alpha_1^2\right]W_1^{(L),1}(x)
\\=8h\alpha_1xy_{(L),1}^2\frac{\mathrm{d}}{\mathrm{d}x}W_1^{(L),0}(x)+4h\alpha_1\left[x^2-6(\alpha_1+1)x+8\alpha_1^2\right]W_1^{(L),0}(x)
\\+2h\left[x(x-1)+4\alpha_1x+8\alpha_1^2\right],
\end{multline}

\begin{multline} \label{eq:rr42}
xy_{(L),1}^4\frac{\mathrm{d}}{\mathrm{d}x}W_1^{(L),2}(x)-2y_{(L),1}^2\left[(\alpha_1+1)x-2\alpha_1^2\right]W_1^{(L),2}(x)
\\=8h\alpha_1xy_{(L),1}^2\frac{\mathrm{d}}{\mathrm{d}x}W_1^{(L),1}(x)+4h\alpha_1\left[x^2-6(\alpha_1+1)x+8\alpha_1^2\right]W_1^{(L),1}(x)
\\+\tfrac{5}{2}x^3y_{(L),1}^2\frac{\mathrm{d}^3}{\mathrm{d}x^3}W_1^{(L),0}(x)+x^2\left[8x^2-38(\alpha_1+1)x+44\alpha_1^2\right]\frac{\mathrm{d}^2}{\mathrm{d}x^2}W_1^{(L),0}(x)
\\+2x\left[x^2-6(\alpha_1+1)x+10\alpha_1^2\right]\frac{\mathrm{d}}{\mathrm{d}x}W_1^{(L),0}(x)
\\+\left[4(\alpha_1+1)x-8\alpha_1^2\right]W_1^{(L),0}(x)-2(2\alpha_1+x),
\end{multline}
and for $l\geq3$, the general differential equation
\begin{multline} \label{eq:rr43}
xy_{(L),1}^4\frac{\mathrm{d}}{\mathrm{d}x}W_1^{(L),l}(x)-2y_{(L),1}^2\left[(\alpha_1+1)x-2\alpha_1^2\right]W_1^{(L),l}(x)
\\=8h\alpha_1xy_{(L),1}^2\frac{\mathrm{d}}{\mathrm{d}x}W_1^{(L),l-1}(x)+4h\alpha_1\left[x^2-6(\alpha_1+1)x+8\alpha_1^2\right]W_1^{(L),l-1}(x)
\\+\tfrac{5}{2}x^3y_{(L),1}^2\frac{\mathrm{d}^3}{\mathrm{d}x^3}W_1^{(L),l-2}(x)+x^2\left[8x^2-38(\alpha_1+1)x+44\alpha_1^2\right]\frac{\mathrm{d}^2}{\mathrm{d}x^2}W_1^{(L),l-2}(x)
\\+2x\left[x^2-6(\alpha_1+1)x+10\alpha_1^2\right]\frac{\mathrm{d}}{\mathrm{d}x}W_1^{(L),l-2}(x)+\left[4(\alpha_1+1)x-8\alpha_1^2\right]W_1^{(L),l-2}(x)
\\-2h\alpha_1x\left[5x^2\frac{\mathrm{d}^3}{\mathrm{d}x^3}+22x\frac{\mathrm{d}^2}{\mathrm{d}x^2}+14\frac{\mathrm{d}}{\mathrm{d}x}\right]W_1^{(L),l-3}(x)
\\-\left[x^5\frac{\mathrm{d}^5}{\mathrm{d}x^5}+10x^4\frac{\mathrm{d}^4}{\mathrm{d}x^4}+22x^3\frac{\mathrm{d}^3}{\mathrm{d}x^3}+4x^2\frac{\mathrm{d}^2}{\mathrm{d}x^2}-4x\frac{\mathrm{d}}{\mathrm{d}x}\right]W_1^{(L),l-4}(x),
\end{multline}
where we set $W_1^{(L),-1}:=0$.
\end{proposition}

\begin{proposition} \label{P3.14}
The expansion coefficients of the LSE scaled density's resolvent $\tilde{W}_{4,N}^{(L)}(x)$ satisfy the differential equations presented in Proposition \ref{P3.13} upon replacing $\alpha_1$ with $\alpha_1/4$.
\end{proposition}

Finally, the analogous proposition for the JUE is as follows:

\begin{proposition} \label{P3.15}
The expansion coefficients of the JUE scaled density's resolvent $\tilde{W}_{2,N}^{(J)}(x)$ \eqref{eq:rr30} satisfy the differential equation
\begin{multline} \label{eq:rr44}
\left[(\alpha_1+\alpha_2+2)^2x^2-2(\alpha_1+2)(\alpha_1+\alpha_2)x+\alpha_1^2\right]x(x-1)\frac{\mathrm{d}}{\mathrm{d}x}W_2^{(J),0}(x)
\\+\Big[\alpha_1^2(x-1)^3+\alpha_1(\alpha_2+2)x(1-x)(1-2x)+(\alpha_2+2)^2x^3
\\-2(\alpha_2+1)(3x^2+x)\Big]W_2^{(J),0}(x)=(\alpha_1+\alpha_2+1)(\alpha_1(1-x)+\alpha_2x),
\end{multline}
and for $l\geq1$,
\begin{multline} \label{eq:rr45}
\left[(\alpha_1+\alpha_2+2)^2x^2-2(\alpha_1+2)(\alpha_1+\alpha_2)x+\alpha_1^2\right]x(x-1)\frac{\mathrm{d}}{\mathrm{d}x}W_2^{(J),l}(x)
\\+\Big[\alpha_1^2(x-1)^3+\alpha_1(\alpha_2+2)x(1-x)(1-2x)+(\alpha_2+2)^2x^3
\\-2(\alpha_2+1)(3x^2+x)\Big]W_2^{(J),l}(x)=x^3(x-1)^3\frac{\mathrm{d}^3}{\mathrm{d}x^3}W_2^{(J),l-2}(x)
\\-4x^2(x-1)^2(1-2x)\frac{\mathrm{d}^2}{\mathrm{d}x^2}W_2^{(J),l-2}(x)+2x(x-1)(7x(x-1)+1)\frac{\mathrm{d}}{\mathrm{d}x}W_2^{(J),l-2}(x)
\\-2x(x-1)(1-2x)W_2^{(J),l-2}(x),
\end{multline}
where we set $W_2^{(J),-1}(x)=0$.
\end{proposition}
Similar to Proposition \ref{P3.12}, in the $a=\alpha_1N,\,b=\alpha_2N$ setting, $W_2^{(J),k}$ vanishes for odd $k$, and the differential equations of Proposition \ref{P3.15} constitute a recursion over even $l$. As in \S\ref{s3.2}, we do not present the analogous differential equations for the JOE and JSE due to their cumbersome structure.

\begin{remark} \label{R3.16}
Comparing \eqref{eq:rr31}, \eqref{eq:rr34}, and the analogous result in \cite{HT12} to each other shows that the differential equation for $W_{\beta}^{(G),0}(x)$ is independent of $\beta$. This can also be observed for $W_{\beta}^{(L),0}(x)$ when comparing \eqref{eq:rr38} to \eqref{eq:rr40} while rescaling $\alpha_1$ by a factor of $1/\kappa$. Furthermore, one can readily check from Theorem \ref{T2.8} that differential equation \eqref{eq:rr44} is valid when replacing $W_2^{(J),0}(x)$ with $W_1^{(J),0}(x)$ or $W_4^{(J),0}(x)$, granted that $\alpha_1$ and $\alpha_2$ are rescaled appropriately. When we take into account the requirement that $W_{\beta}^{0}(x)\underset{x\rightarrow\infty}{\sim}1/x$, these differential equations and the spectral curves seen in topological recursion both yield the same unique $\beta$-independent solutions. This is in keeping with the fact that the large $N$ limiting forms of $\tilde{\rho}_{(1),\beta,N}(x)$ (the inverse Stieltjes transforms of $W_{\beta}^0(x)$) in these cases are independent of $\beta$, as discussed earlier in this subsection.
\end{remark}

\setcounter{equation}{0}
\section{Differential Equations for the Soft and Hard Edge Scaled Densities}\label{s4}
In this section, we present the second of our promised applications of the differential equations in Section \ref{s2}. Namely, we derive differential equations satisfied by the eigenvalue densities when they have been centred on either the largest or smallest eigenvalue and they have been scaled so that the mean spacing between this eigenvalue and its neighbour is order unity. Upon recentring and scaling in this manner, the eigenvalue densities considered fall into one of two universal classes: The limiting smoothed density exhibits either a square root profile, which we call a soft edge, or it exhibits an inverse square root singularity, which we call a hard edge (see e.g. \cite{Fo12,WF12}). In particular, the Gaussian ensemble eigenvalue density exhibits only soft edges in accordance with the Wigner semi-circle law. Likewise, the Laguerre ensemble eigenvalue density also has a soft edge at the largest eigenvalue. Both edges of the Jacobi density and the left edge of the Laguerre density behave either as hard edges when the corresponding parameter is of order unity -- $a$ for the regime of the smallest eigenvalue and $b$ for the largest eigenvalue -- or as soft edges if said parameter is of order $N$. A consequence of universality is that for each $\beta$ value, all soft edges are statistically equivalent, and likewise for hard edges. Thus, we specify the differential equations satisfied by the eigenvalue densities with soft and hard edge scalings for $\beta=1,2$, and $4$, and with soft edge scaling for $\beta=2/3$ and $6$. The universal soft and hard edge limiting densities are to be denoted as $\rho_{(1),\beta,\infty}^{(soft)}(x)$ and $\rho_{(1),\beta,\infty}^{(hard)}(x)$, respectively.

\begin{theorem} \label{T4.1}
Define the soft edge limiting forms of the differential operators introduced in Section \ref{s2} as
\begin{align} \label{eq:scale1}
\mathcal{D}_{\beta,\infty}^{(soft)} =
\begin{cases}
\frac{\mathrm{d}^3}{\mathrm{d}x^3}-4x\frac{\mathrm{d}}{\mathrm{d}x}+2,&\beta=2,
\\ \frac{\mathrm{d}^5}{\mathrm{d}x^5}-10\kappa x\frac{\mathrm{d}^3}{\mathrm{d}x^3}+6\kappa\frac{\mathrm{d}^2}{\mathrm{d}x^2}+16\kappa^2x^2\frac{\mathrm{d}}{\mathrm{d}x}-8\kappa^2x,&\beta=1,4,
\\ 3\frac{\mathrm{d}^7}{\mathrm{d}x^7}-56\kappa x\frac{\mathrm{d}^5}{\mathrm{d}x^5}+28\kappa\frac{\mathrm{d}^4}{\mathrm{d}x^4}+\frac{784}{3}\kappa^2x^2\frac{\mathrm{d}^3}{\mathrm{d}x^3}
\\ \quad-208\kappa^2x\frac{\mathrm{d}^2}{\mathrm{d}x^2}-4\kappa^2(64\kappa x^3-17)\frac{\mathrm{d}}{\mathrm{d}x}+128\kappa^3x^2,&\beta=2/3,6,
\end{cases}
\end{align}
where we recall that $\kappa=\beta/2$. Then for $\beta\in\{2/3,1,2,4,6\}$, at leading order, the Gaussian, Laguerre, and Jacobi eigenvalue densities satisfy the following differential equation in the soft edge limit:
\begin{equation} \label{eq:scale2}
\mathcal{D}_{\beta,\infty}^{(soft)}\,\rho_{(1),\beta,\infty}^{(soft)}(x) = 0.
\end{equation}
\end{theorem}
\begin{proof}
Make the change of variables \cite{Fo93d} $x\mapsto\sqrt{\kappa}\left(\sqrt{2N}+\frac{x}{\sqrt{2}N^{1/6}}\right)$ in \eqref{eq:de2} and \eqref{eq:de38}. Then, multiply through by $N^{-1/2}$ for $\beta=2$, $N^{-5/6}$ for $\beta=1$ and $4$, or $N^{-7/6}$ for $\beta=2/3$ and $6$. Equating terms of order one then yields \eqref{eq:scale2} above, while all other terms vanish in the $N\rightarrow\infty$ limit.
\end{proof}
In the case $\beta=2$, the differential equation \eqref{eq:scale2} satisfied by the soft edge density has been isolated in the earlier works
of Brack et al.~\cite[Eq.~(C.2) with $\hbar^2/m=1$,
$\lambda_M - V(r) = - {1 \over 2} r$, $D=1$]{BKMR10} and of Dean et al. \cite[Eq.~(207) with $d=1$]{DDMS16}. For $\beta=1,2$, and $4$, the above proof can be replicated by instead considering the differential equations \eqref{eq:de19}, \eqref{eq:de24}, and \eqref{eq:de29} for the Jacobi and Laguerre eigenvalue densities due to universality. Explicitly, differential equation \eqref{eq:scale2} is satisfied by the leading order term of the following scaled densities in the large $N$ limit \cite{FT18,Jo01,Fo12,BF98}:
\begin{itemize}
\item In the regime of the largest eigenvalue,
\begin{equation}\label{4.2a}
\rho_{(1),\beta,N}^{(G)}\left(\sqrt{\kappa}\left(\sqrt{2N}+\delta_G+\frac{x}{\sqrt{2}N^{1/6}}\right)\right),
\end{equation}
where $\delta_G={\rm o}(N^{-1/6})$ is an arbitrary parameter (the regime of the smallest eigenvalue can be treated by exploiting the symmetry $x\mapsto-x$);
\item In the regime of the largest eigenvalue with $a={\rm O}(1)$,
\begin{equation}\label{4.2b}
\rho_{(1),\beta,N}^{(L)}\left(\kappa\left(4N+\delta_L^{(l,a={\rm O}(1))}+2(2N)^{1/3}x\right)\right),
\end{equation}
where $\delta_L^{(\,\cdot\,)}={\rm o}(N^{1/3})$ is henceforth an arbitrary parameter;
\item In the regime of the largest eigenvalue with $a=\alpha_1N$ and $\alpha_1={\rm O}(1)$,
\begin{equation}\label{4.2c}
\rho_{(1),\beta,N}^{(L)}\left(\kappa\left(q_+^2N+\delta_L^{(l,a={\rm O}(N))}+\frac{q_+}{q_+-1}(q_+N)^{1/3}x\right)\right),
\end{equation}
where we define $q_{\pm}:=\sqrt{1+\frac{\alpha_1}{\kappa}}\pm1$ so that $q_{\pm}^2$ are the endpoints of the support of the limiting smoothed density $\tilde{\rho}_{(1),\beta,N}^{(L)}$ as defined by \eqref{eq:rr24} (note that $q_+\rightarrow2$ as $\alpha_1\rightarrow0$, so we have consistency with the scaling given above for the $a={\rm O}(1)$ regime);
\item In the regime of the smallest eigenvalue with $a=\alpha_1N$ and $\alpha_1={\rm O}(1)$,
\begin{equation}\label{4.2d}
\rho_{(1),\beta,N}^{(L)}\left(\kappa\left(q_-^2N-\delta_L^{(s,a={\rm O}(N))}-q_-\left(\frac{q_-N}{q_-+1}\right)^{1/3}x\right)\right);
\end{equation}
\item In the regime of either the largest or smallest eigenvalue, with either $a=\alpha_1N$ and $\alpha_1,b={\rm O}(1)$, or $a=\alpha_1N$, $b=\alpha_2N$ and $\alpha_1,\alpha_2={\rm O}(1)$, $\rho_{(1),\beta,N}^{(J)}$ scaled according to \cite{HF12,Jo08}.
\end{itemize}

\begin{remark}
For even $\beta$, $\rho_{(1),\beta,\infty}^{(soft)}(x)$ has an explicit representation as a $\beta$-dimensional integral due to \cite{DF06}, while \cite{FFG06} provides alternate forms for $\beta=1,2$, and $4$. To date, there is no explicit functional form for $\rho_{(1),2/3,\infty}^{(soft)}(x)$. On the other hand, \cite{DV13} shows for all $\beta>0$ that for the first two leading orders as $x\rightarrow\infty$,
\begin{equation*}
\rho_{(1),\beta,\infty}^{(soft)}(x)\propto\frac{\exp\left(-4\kappa x^{3/2}/3\right)}{x^{3\kappa/2}},
\end{equation*}
which is an extension of the even-$\beta$ result of \cite{Fo12a},
\begin{equation} \label{eq:scale3}
\rho_{(1),\beta,\infty}^{(soft)}(x)\underset{x\rightarrow\infty}{\sim}\frac{1}{\pi}\frac{\Gamma(1+\kappa)}{(8\kappa)^{\kappa}}\frac{\exp\left(-4\kappa x^{3/2}/3\right)}{x^{3\kappa/2}}.
\end{equation}
This result is consistent with the differential equations of Theorem \ref{T4.1}. So too is the result \cite{DF06},
\begin{equation} \label{eq:scale4}
\rho_{(1),\beta,\infty}^{(soft)}(x)\underset{x\rightarrow-\infty}{\sim}\frac{\sqrt{|x|}}{\pi},\quad\beta\in2\mathbb{N}.
\end{equation}
Likewise, Theorem \ref{T4.4} below is consistent with the result \cite{Fo93c},
\begin{equation} \label{eq:scale5}
\rho_{(1),\beta,\infty}^{(hard)}(x)\underset{x\rightarrow\infty}{\sim}\frac{1}{2\pi\sqrt{x}},\quad\beta\in2\mathbb{N}.
\end{equation}
Note that the asymptotic forms \eqref{eq:scale3}, \eqref{eq:scale4}, and \eqref{eq:scale5} respectively capture the facts that moving past the soft edge results in exponential decay, moving from the soft edge into the bulk results in a square root profile, and moving from the hard edge into the bulk shows a square root singularity.
\end{remark}

\begin{remark}
It has previously been observed in \cite{FT18} that the differential equation characterisation of the soft edge scaled density can be extended to similarly characterise the optimal leading order correction term. The latter is obtained by a tuning of $\delta_G$ and $\delta_L^{(\,\cdot\,)}$ in (\ref{4.2a})--(\ref{4.2d}) so as to obtain the fastest possible decay in $N$ of the leading order correction, and thus the fastest possible convergence to the limit. On this latter point, and considering the Gaussian case for definiteness, we know from \cite{FT19} that for $\beta$ even (at least),
\begin{align*}
\mu_{\beta,N}(x)  :&= {\sqrt{\kappa} \over \sqrt{2} N^{7/6}} 
\rho_{(1),\beta,N}^{(G)}\left(\sqrt{\kappa}\left(\sqrt{2N}+\delta_G+\frac{x}{\sqrt{2}N^{1/6}}\right)\right)  \\ &=
\rho_{(1),\beta,\infty}^{(soft)}(x) + \Big (
\sqrt{2} N^{1/6} \delta_G - (1 - 1/\kappa)/(2 N^{1/3}) \Big )
{\mathrm{d} \over \mathrm{d}x} \rho_{(1),\beta,\infty}^{(soft)}(x) 
\\&\quad+ {\rm O}\Big ( {1 \over N^{2/3}} \Big ).
\end{align*}
Thus, choosing $\delta_G =  (1 - 1/\kappa)/(2 \sqrt{2N})$ gives the fastest
convergence to the limit,
$$
\mu_{\beta,N}(x) = \rho_{(1),\beta,\infty}^{(soft)}(x) + {1 \over N^{2/3}} \hat{\mu}_{\beta}(x) +
{\rm o}(N^{-2/3})
$$
for some $\hat{\mu}_\beta(x)$ which, for the values of $\beta$ permitting a differential equation
characterisation of $\rho_{(1),\beta,N}^{(G)}$ and $\rho_{(1),\beta,\infty}^{(soft)}(x)$, can itself be characterised as the solution
of a differential equation. The simplest case is $\beta = 2$, when
\begin{equation} \label{4.12a}
\hat{\mu}_2'''(x) - 4 x \hat{\mu}_2'(x) + 2 \hat{\mu}_2(x) = x^2 {\mathrm{d}\over \mathrm{d}x} \rho_{(1),2,\infty}^{(soft)}(x) - x \rho_{(1),2,\infty}^{(soft)}(x),
\end{equation}
which is an inhomogeneous generalisation of (\ref{eq:scale2}) for $\beta = 2$. For the particular Laguerre soft edge scaling \eqref{4.2c}, again with $\beta=2$, the analogue of \eqref{4.12a} is given in \cite[Eq.~(4.19)]{FT18}. For all other even $\beta$, \cite{FT19} shows that scalings \eqref{4.2b} and \eqref{4.2c} become optimal when we set
\begin{equation*}
\delta_L^{(l,a={\rm O}(1))}=2a/\kappa\quad\textrm{and}\quad\delta_L^{(l,a={\rm O}(N))}=\left(1-\frac{1}{\kappa}\right)\frac{\alpha_1/\kappa}{2\sqrt{\alpha_1/\kappa+1}}.
\end{equation*}
Applying these scalings to differential equation \eqref{eq:de29} shows that they remain optimal in the $\beta=1$ case when choosing $\delta_L^{(\,\cdot\,)}$ as above. Similarly, when $\beta=1,2$ or $4$, scaling \eqref{4.2d} becomes optimal when we take $\delta_L^{(s,a={\rm O}(N))}$ to be equal to the optimal choice of $\delta_L^{(l,a={\rm O}(N))}$ given above.
\end{remark}

At the hard edge, in contrast to the soft edge scaling, the dependence on the
exponent $a$ in the Laguerre or Jacobi weight remains.

\begin{theorem} \label{T4.4}
Define the hard edge limiting forms of the differential operators introduced in Section \ref{s2} as
\begin{align} \label{eq:scale7}
\mathcal{D}_{\beta,\infty}^{(hard)} =
\begin{cases}
x^3\frac{\mathrm{d}^3}{\mathrm{d}x^3}+4x^2\frac{\mathrm{d}^2}{\mathrm{d}x^2}+\left[x-a^2+2\right]x\frac{\mathrm{d}}{\mathrm{d}x}+\tfrac{1}{2}x-a^2,&\beta=2,
\\ 4x^5\frac{\mathrm{d}^5}{\mathrm{d}x^5}+40x^4\frac{\mathrm{d}^4}{\mathrm{d}x^4}+\left[10\kappa x-5\tilde{a}+88\right]x^3\frac{\mathrm{d}^3}{\mathrm{d}x^3}&
\\\quad+\left[38\kappa x-22\tilde{a}+16\right]x^2\frac{\mathrm{d}^2}{\mathrm{d}x^2}&
\\\quad+\left[\left(2\kappa x-\tilde{a}\right)^2+12\kappa x-14\tilde{a}-16\right]x\frac{\mathrm{d}}{\mathrm{d}x}&
\\\quad+(2\kappa x-\tilde{a})(\kappa x-\tilde{a})-4\kappa x,&\beta=1,4,
\end{cases}
\end{align}
where we retain the definition of $\tilde{a}$ given in Theorem \ref{T2.8}. Then for $\beta\in\{1,2,4\}$, the hard edge scaled densities satisfy the differential equations
\begin{equation} \label{eq:scale8}
\mathcal{D}_{\beta,\infty}^{(hard)}\,\rho_{(1),\beta,\infty}^{(hard)}(x) = 0.
\end{equation}
\end{theorem}
\begin{proof}
Since the Gaussian ensemble eigenvalue density does not exhibit a hard edge, we turn to the Laguerre ensemble as it is the next-simplest to work with. Thus, we begin by changing variables \cite{NF95} $x\mapsto\kappa x/(4N)$ in \eqref{eq:de29}. Equating terms of order one yields the differential equation \eqref{eq:scale8} above, and we note that all other terms are ${\rm O}(\frac{1}{N})$.
\end{proof}

\begin{remark}
Our proof made use of Laguerre ensemble finite $N$ differential equations. One can check that applying the following hard edge scalings to the differential equations \eqref{eq:de19} and \eqref{eq:de24} reclaims Theorem \ref{T4.4}:
\begin{itemize}
\item In the regime of the smallest eigenvalue with $a,b={\rm O}(1)$ \cite{NF95},
\begin{equation*}
\rho_{(1),\beta,N}^{(J)}\left(\frac{x}{4N^2}\right);
\end{equation*}
\item In the regime of the smallest eigenvalue with $b=\alpha_2N$ and $a,\alpha_2={\rm O}(1)$,
\begin{equation*}
\rho_{(1),\beta,N}^{(J)}\left(\frac{x}{4(1+\alpha_2/\kappa)N^2}\right);
\end{equation*}
\item Before scaling, we may exploit the symmetry $(x,a,b)\leftrightarrow(1-x,b,a)$ for the Jacobi ensemble to characterise the hard edge scalings of the Jacobi ensemble at the largest eigenvalue.
\end{itemize}
\end{remark}

\begin{remark} The optimal hard edge scaling variable
	for the Laguerre ensemble has been identified in \cite{FT19a} as
	$x \mapsto \kappa x /(4N + 2a/\kappa)$; the leading correction term to
	the limiting density is then ${\rm O}(1/N^2)$, as distinct from the slower decaying
	${\rm O}(1/N^{2/3})$ optimal correction at the soft edge.
	\end{remark}

\section*{Acknowledgments}
The authors are grateful to the referees for their valuable feedback. The work of PJF was partially supported by the Australian Research Council Grant DP170102028 and the ARC Centre of Excellence for Mathematical and Statistical Frontiers, and that of AAR by the Australian Government Research Training Program Scholarship and the ARC Centre of Excellence for Mathematical and Statistical Frontiers.

\setcounter{equation}{0}
\appendix
\section{Particular Stieltjes Transforms} \label{A}
To take the Stieltjes transforms of equations \eqref{eq:de19} and \eqref{eq:de24}, we need to compute terms of the form
\begin{align*}
\int_0^1\frac{x^p(1-x)^q}{s-x}\frac{\mathrm{d}^n}{\mathrm{d}x^n}\rho_{(1),\beta,N}^{(J)}(x)\,\mathrm{d}x
\end{align*}
for $\beta=2$ or $4$, $0\leq n\leq5$, $n\leq q\leq n+2$, and $q\leq p\leq q+1$ all integers. To this end, we define
\begin{equation}\label{eq:A1}
\mathcal{I}_{\beta}(s;p,q,n,k):=\int_0^1\frac{x^p(1-x)^q}{(s-x)^k}\frac{\mathrm{d}^n}{\mathrm{d}x^n}\rho_{(1),\beta,N}^{(J)}(x)\,\mathrm{d}x
\end{equation}
for integers $0\leq n\leq q\leq p$ and $k\geq0$. Then, integration by parts gives the identity
\begin{align}
\mathcal{I}_{\beta}(s;p,q,n,k)&=(p+q)\mathcal{I}_{\beta}(s;p,q-1,n-1,k)-p\mathcal{I}_{\beta}(s;p-1,q-1,n-1,k)\nonumber
\\&\quad-k\mathcal{I}_{\beta}(s;p,q,n-1,k+1).
\end{align}
Applying this identity $n$ times allows us to reduce $\mathcal{I}_{\beta}(s;p,q,n,k)$ to an expression involving terms of the form $\mathcal{I}_{\beta}(s;p,q,0,k)$. Then, considering $(s-x)^{-k-1}=(-1)^k/k!\,\partial_s^k(s-x)^{-1}$ for $k\geq0$ gives us
\begin{align}
\mathcal{I}_{\beta}(s;p,q,0,k+1)=\frac{(-1)^k}{k!}\frac{\mathrm{d}^k}{\mathrm{d}s^k}\mathcal{I}_{\beta}(s;p,q,0,1),\quad k\geq0,
\end{align}
which allows us to further reduce to an expression involving terms of the form $\mathcal{I}_{\beta}(s;p,q,0,1)$. Finally, factorisation of $x^p-s^p$ for positive integer $p$ yields
\begin{align} \label{eq:A4}
\mathcal{I}_{\beta}(s;p,q,0,1)=s^p\mathcal{I}_{\beta}(s;0,q,0,1)-\sum_{l=0}^{p-1}s^{p-l-1}\mathcal{I}_{\beta}(s;l,q,0,0),\quad p\geq1
\end{align}
and likewise
\begin{align}
\mathcal{I}_{\beta}(s;0,q,0,1)=\sum_{m=0}^q\binom{q}{m}(-1)^m\left[s^mW_{\beta,N}^{(J)}(s)-\sum_{l=0}^{m-1}s^{m-l-1}m_l^{(J)}\right],\quad q\geq1,
\end{align}
where $m_l^{(J)}$ is the $l\textsuperscript{th}$ moment of $\rho_{(1),\beta,N}^{(J)}(x)$. These last two equations thus reduce our expression to one involving powers of $s$, derivatives of $W_{\beta,N}^{(J)}(s)$, and moments of $\rho_{(1),\beta,N}^{(J)}(x)$.

To begin, we list
\begin{align*}
\mathcal{I}_{\beta}(s;1,1,1,1)&=s(1-s)\frac{\mathrm{d}}{\mathrm{d}s}W_{\beta,N}^{(J)}(s)-N,
\\\mathcal{I}_{\beta}(s;2,2,2,1)&=s^2(1-s)^2\frac{\mathrm{d}^2}{\mathrm{d}s^2}W_{\beta,N}^{(J)}(s)-2N(s-2)-6m_1^{(J)},
\\\mathcal{I}_{\beta}(s;3,3,3,1)&=s^3(1-s)^3\frac{\mathrm{d}^3}{\mathrm{d}s^3}W_{\beta,N}^{(J)}(s)-6N\left(s^2-3s+3\right)
\\&\quad-24m_1^{(J)}(s-3)-60m_2^{(J)},
\\\mathcal{I}_{\beta}(s;4,4,4,1)&=s^4(1-s)^4\frac{\mathrm{d}^4}{\mathrm{d}s^4}W_{\beta,N}^{(J)}(s)-24N(s^3-4s^2+6s-4)
\\&\qquad-120m_1^{(J)}\left(s^2-4s+6\right)-360m_2^{(J)}(s-4)-840m_3^{(J)},
\\\mathcal{I}_{\beta}(s;5,5,5,1)&=s^5(1-s)^5\frac{\mathrm{d}^5}{\mathrm{d}s^5}W_{\beta,N}^{(J)}(s)-120N(s^4-5s^3+10s^2-10s+5)
\\&\qquad-720m_1^{(J)}(s^3-5s^2+10s-10)-2520m_2^{(J)}(s^2-5s+10)
\\&\qquad-6720m_3^{(J)}(s-5)-15120m_4^{(J)}.
\end{align*}
All of the necessary $\mathcal{I}_{\beta}(s;p,q,n,k)$ can be obtained from these through variants of identity \eqref{eq:A4} and integration by parts. For example,
\begin{multline*}
\mathcal{I}_{\beta}(s;n+1,n,n,1)=s\mathcal{I}_{\beta}(s;n,n,n,1)-\mathcal{I}_{\beta}(s;n,n,n,0)
\\=s\mathcal{I}_{\beta}(s;n,n,n,1)+(-1)^{n+1}\int_0^1\frac{\mathrm{d}^n}{\mathrm{d}x^n}(x^n(1-x)^n)\rho_{(1),\beta,N}^{(J)}(x)\,\mathrm{d}x,
\end{multline*}
which only requires knowledge of the moments $m_0^{(J)}$ to $m_n^{(J)}$ of $\rho_{(1),\beta,N}^{(J)}(x)$, in addition to a term from the above list. It should be noted that applying the Stieltjes transform to $x^p(1-x)^q\frac{\mathrm{d}^n}{\mathrm{d}x^n}\rho_{(1),\beta,N}^{(J)}(x)$ produces $s^p(1-s)^q\frac{\mathrm{d}^n}{\mathrm{d}s^n}W_{\beta,N}^{(J)}(s)$ plus terms that do not involve $W_{\beta,N}^{(J)}(s)$. It follows that the differential equations for the resolvent will be the same as those for the density, with additional inhomogeneous terms.

\end{document}